\documentclass[final,times,3p,sort&compress]{elsarticle}
\pdfoutput=1
\usepackage[pagewise]{lineno}
\usepackage{hyphenat}

\usepackage{microtype}
\usepackage[neverdecrease]{paralist}

\usepackage{amsmath,amssymb}
\usepackage[text size=scriptsize]{todonotes}
\usepackage{mathtools}

\usepackage[colorlinks=true,breaklinks=true,bookmarks=true,urlcolor=blue,citecolor=blue,linkcolor=blue,bookmarksopen=false,draft=false]{hyperref}
\usepackage[nameinlink,capitalize]{cleveref}
\usepackage{doi}
\usepackage[section]{algorithm}
\hypersetup{colorlinks,allcolors=blue}

\DeclareMathOperator*{\argmax}{arg\,max}

\newcommand{\I}{\mathcal{I}}
\newcommand{\alp}{\alpha}
\newcommand{\bet}{\beta}
\newcommand{\gam}{\gamma}
\newcommand{\ka}{k}
\newcommand{\cF}{\mathcal{F}}
\newcommand{\cH}{\mathcal{H}}
\newcommand{\cFs}{\widehat{\mathcal{F}}}
\newcommand{\cS}{\mathcal S}
\newcommand{\cB}{\mathcal B}
\newcommand{\cSs}{{\widehat{\mathcal S}}}
\newcommand{\cSst}{{\mathcal R}}
\newcommand{\cR}{\mathcal R}

\newcommand{\dw}{\mathbf{w}}
\newcommand{\dg}{\mathbf{g}}
\newcommand{\dtext}{inductive union maximizing function}
\newcommand{\dtex}{inductive union maximizing}
\newcommand{\dterm}{generating}
\newcommand{\Dtext}{Inductive union maximizing function}

\DeclareMathOperator{\col}{col}
\newcommand{\familyS}{\cS}
\newcommand{\N}{\mathbb N}
\newcommand{\Z}{\mathbb Z}
\newcommand{\R}{\mathbb R}
\newcommand{\F}{\mathbb F}
\newcommand{\bigO}{O}
\newcommand{\FPT}{\text{FPT}}
\newcommand{\Wone}{W[1]}
\newcommand{\Wtwo}{W[2]}
\DeclareMathOperator{\poly}{poly}
\DeclareMathOperator{\polylog}{polylog}

\newcommand{\mdmc}{UFLP-MC}
\newcommand{\spmc}{SPMC}
\newcommand{\mdmcs}{UFLP-MC}
\newcommand{\mdmcc}{UFLP-MCC}

\newcommand{\clique}{\textsc{Clique}}

\newtheorem{theorem}{Theorem}
\numberwithin{theorem}{section}
\newtheorem{remark}[theorem]{Remark}
\newtheorem{corollary}[theorem]{Corollary}
\newtheorem{lemma}[theorem]{Lemma}
\newtheorem{proposition}[theorem]{Proposition}
\newtheorem{definition}[theorem]{Definition}
\newtheorem{construction}[theorem]{Construction}
\newtheorem{problem}[theorem]{Problem}
\newtheorem{example}[theorem]{Example}
\newtheorem{observation}[theorem]{Observation}

\crefname{line}{line}{lines}
\crefname{lemma}{Lemma}{Lemmas}
\crefname{problem}{Problem}{Problems}
\crefname{remarkn}{Remark}{Remarks}
\crefname{proposition}{Proposition}{Propositions}

\crefname{example}{Example}{Examples}

\newcommand{\fieldconst}[1]{
	$\mathbb F = \mathbb F_{p^d}$ such that  $p$ is a prime polynomially upper\hyp bounded in #1}

\newproof{proof}{Proof}

\journal{Discrete Applied Mathematics}

\begin{document}

\begin{frontmatter}

  \title{Representative families for matroid intersections,
    with applications to\\location,
    packing, and covering problems\tnoteref{t1}}
      \tnotetext[t1]{The results in this work
        were announced in an extended abstract at CIAC 2019 \citep{BTZ19}.}

\author[1]{René van Bevern\corref{cor1}
}\ead{rvb@nsu.ru}

\author[1,2]{Oxana Yu.\ Tsidulko
}
\ead{tsidulko@math.nsc.ru}

\author[3]{Philipp Zschoche
}
\ead{zschoche@tu-berlin.de}

\cortext[cor1]{Correspondence to: Department of Mechanics and Mathematics,
  Novosibirsk State University,
  ul.\ Pirogova 1,
  Novosibirsk, 630090, Russian Federation}
\address[1]{Department of Mechanics and Mathematics, Novosibirsk State University, Novosibirsk, Russian Federation}
\address[2]{Sobolev Institute of Mathematics of the Siberian Branch of the Russian Academy of Sciences, Novosibirsk, Russian Federation}
\address[3]{Technische Universität Berlin, Faculty IV, Algorithmics and Computational Complexity, Germany}

\begin{abstract}
  We show algorithms
  for computing representative families
  for matroid intersections
  and use them in fixed\hyp parameter algorithms
  for set packing, set covering, and facility location problems
  with multiple matroid constraints.
  We
  complement our tractability results
  by hardness results.
\end{abstract}


\end{frontmatter}
\thispagestyle{empty}

\section{Introduction}
\label{sec:introduction}


\noindent
  Matroids are
  an important tool 
in the development
of fixed\hyp parameter algorithms~\citep{PS16}
and many of these algorithms are based
on the fast construction of
so\hyp called \emph{representative families} 
\citep{Mar09,FLPS16,FLPS17,LMPS18,GMPZ15,Zehavi15}.
Generalizing these,
we present algorithms to compute representative families
not for single matroids, but for
intersections of multiple matroids.

\looseness=-1
Using this generalization,
we derive fixed\hyp parameter algorithms
for packing, covering and facility location
problems with multiple matroid constraints.
Herein,
our algorithms for packing and covering problems
generalize and unify several fixed\hyp parameter algorithms
for covering problems
known in the literature \citep{Mar08,Mar09,BPS16}.
In the context of facility location problems,
matroid constraints
can model natural facility location scenarios like
``open at most $k_i$~facilities of type~$i$'' \citep{CLLW16},
even if types are not disjoint,
moving facilities \citep{Swa16},
but also problems in social network
analysis \citep[Section~3]{BTZ19}.

\paragraph{Organization of this work}
In \cref{sec:preliminaries}
we provide basic definitions
from parameterized complexity and
matroid theory.
\cref{sec:tools} presents
our algorithms for constructing representative families
for matroid intersections.
\cref{sec:spmc,sec:algorithm} present
our fixed\hyp parameter algorithm for set packing
and facility location problems
with multiple matroid constraints, respectively.
Related work and context for each result
is provided in the respective subsections.

\section{Preliminaries}

\label{sec:preliminaries}


\subsection{Parameterized complexity}
\noindent
The main idea of fixed\hyp parameter algorithms
is to accept the exponential running time
seemingly inherent to solving NP-hard problems,
yet to confine the combinatorial explosion
to a parameter of the problem,
which can be small in applications~\citep{CFK+15}.
A problem is \emph{fixed\hyp parameter tractable}
if it can be solved in \(f(k)\cdot \poly(n)\)~time
on inputs of length~\(n\)
and some function~\(f\)
depending only on some parameter~\(k\).
Note that
this requirement is stronger than
an algorithm
that merely runs in polynomial time
for fixed~\(k\),
say, in \(O(n^k)\)~time,
which is inapplicable
even for small values of~\(k\), say $k=10$.
The parameterized analog
of NP and NP-hardness
is the $W$-hierarchy \(\FPT\subseteq\Wone\subseteq\Wtwo\subseteq\dots W[P]\subseteq\text{XP}\) and \(W[t]\)-hardness,
where \FPT{} is the class of
fixed\hyp parameter tractable
decision problems and
all inclusions are conjectured to be strict.
If some \(W[t]\)-hard problem
is in \FPT,
then \(\FPT=W[t]\)~\citep{CFK+15}.

\subsection{Sets and set functions}
\noindent
By \(\N\),
we denote the natural numbers
including zero.
By \(\F_p\),
we denote the field
on \(p\)~elements.
Usually,
we study set systems over a
\emph{universe}~$U=\{1,\dots,n\}$.
By \(A\uplus B\),
we denote the union of sets~\(A\) and~\(B\)
that we require to be disjoint.
By convention,
the intersection of no sets is the whole universe
and the union of no sets is the empty set.

\begin{definition}[partition]
	We call \(Z_1,\dots,Z_\ell\) a \emph{partition} of a set~\(A\)
	if \(Z_1\uplus\dots\uplus Z_\ell=A\)
	and \(Z_i\ne\emptyset\) for each~\(i\in\{1,\dots,\ell\}\).
\end{definition}

\begin{definition}[\boldmath\(\gam\)-family]
	We call \(A\subseteq 2^U\)
	an \emph{\(\gam\)}-family
	if each set in~\(A\) has cardinality
	exactly~\(\gam\).
\end{definition}

\begin{definition}[additive and
	submodular
	set functions]
	A set function \(w\colon 2^U\to\R\)
	is \emph{additive} if,
	for any subsets~\(A\cup B\subseteq U\),
	one has
	\[
	w(A\cup B)=w(A)+w(B)-w(A\cap B).
	\]
	If ``$\leq$'' holds instead of equality,
	then \(w\)~is called \emph{submodular}.
\end{definition}

\begin{definition}[coverage function]
		\label[definition]{def:covfunc}
	A set function
	\[
	w\colon 2^U\to\mathbb N,S\mapsto\Bigl|\bigcup_{u\in S}u\Bigr|,
	\]
	where \(U=2^V\), is a \emph{coverage function}.
\end{definition}
Coverage functions are non\hyp decreasing
and submodular \citep[Section~44.1a]{Sch03}.

\subsection{Matroid fundamentals}
\label{sec:matbas}
\noindent
For proofs of the following propositions
and for illustrative examples
of the following definitions,
we refer to the book by \citet{Oxl92}.

\begin{definition}[matroid]
	\label[definition]{def:matroid}
	A pair $(U,\I)$,
	where $U$~is the \emph{ground set}
	and $\I\subseteq 2^U$ is a family of \emph{independent sets},
	is a \emph{matroid} if the following holds:
	\begin{itemize}
		\item  $\emptyset \in \I$,
		\item  If $A' \subseteq A$ and $A \in \I$, then $A' \in \I$.
		\item  If $A,B \in \I$ and $|A| < |B|$, then there is an $x \in B \setminus A$ such that $A \cup \{x\} \in \I$.
	\end{itemize}
\end{definition}

\begin{definition}[basis, rank]
	An inclusion\hyp wise
	maximal independent set~\(A\in \I\)
	of a matroid~\(M=(U,\I)\)
	is a \emph{basis}.
	The cardinality
	of the bases of~\(M\)
	is
	called the \emph{rank} of~\(M\).
\end{definition}

\begin{definition}[free matroid]
	A \emph{free matroid} is a matroid~\((U,2^U)\)
	in which every set is independent.
\end{definition}


%

\begin{proposition}[matroid union, direct sum] 
	\label[proposition]{def:matunion}
	\label[proposition]{def:union}
	The \emph{union} 
	\[M_1\vee M_2=(U_1\cup U_2, \{J_1\cup J_2\mid J_1\in I_1,J_2\in I_2\})\]
	of two matroids \(M_1=(U_1,I_1)\) and \(M_2=(U_2,I_2)\)
	is a matroid.
	If \(U_1\cap U_2=\emptyset\),
	we write \(M_1\oplus M_2:=M_1\vee M_2\)
	and call their union \emph{direct sum}.
\end{proposition}


\begin{definition}[uniform, partition, and multicolored matroids]
	A \emph{uniform matroid of rank~\(r\)}
	is a matroid~\((U,\I)\)
	such that \(\I:=\{S\subseteq U\mid |S|\leq r\}\).
	The direct sum of uniform matroids
	is called \emph{partition matroid}.
	We call the direct sum of
	uniform matroids
	of rank one
	a \emph{multicolored matroid}.
\end{definition}
Partition matroids are useful
to model constraints
of type ``at most \(k_i\)~items of type~\(i\)''.
\subsection{Matroid representations}
\noindent
In our work,
we will use two different
ways of representing matroids.
The most general representation of matroids
is an \emph{independence oracle},
which in constant time decides
whether a given set
is independent in a given matroid.
One can imagine it
as an algorithm
that decides independence quickly.
We will also use \emph{linear representations}:

\begin{definition}[linear matroids]
  An \emph{\((r\times n)\)-representation} of a matroid~\(M=(U,\I)\)
  over a field~$\F$ is
  a matrix~\(A\in\F^{r\times n}\)
  whose columns are
  labeled by the $n$~elements of~\(U\)
  such that \(S\in \I\)
  if and only if
  the columns of~\(A\)
  with labels in~\(S\)
  are linearly independent over~\(\F\).
  A matroid is \emph{linear}  or \emph{representable over a field \(\F\)}
  if it has a representation over~$\F$.
\end{definition}
One can transform
a representation
of a matroid with rank~\(r\)
over a field~$\F$
into a representation
over~\(\F\) with \(r\)~rows
\citep[Section~2.2]{Oxl92}
and we will always assume
to work with linear representations
of this form.
Not all matroids are representable
over all fields
\citep[Theorem~6.5.4]{Oxl92}.
Some are not representable 
over any field~\citep[Example~1.5.14]{Oxl92}.
If $A_i$ is a $(a_i \times b_i)$-representation 
of a matroid $M_i$ over field $\mathbb F$
for $i \in \{1,\dots, m\}$,
then
a $(\sum_{i=1}^m a_i \times \sum_{i=1}^m b_i)$-representation
of~$\bigoplus_{i=1}^m M_i$ over~\(\mathbb F\)
is computable
in time of
$\bigO(\sum_{i=1}^m a_i \cdot \sum_{i=1}^m b_i)$~operations over~$\mathbb F$
\citep[Exercise~6, p.~132]{Oxl92}.
Uniform matroids
of rank~\(r\)
on a universe of size~\(n\)
are representable over
all fields with at least \(n\)~elements
\citep[Section~3.5]{Mar09}.
The uniform matroid of rank one
is trivially representable
over \emph{all} fields.
Thus,
so are multicolored matroids.


\begin{lemma}
	\label[lemma]{lemma:matroid-set-ext}
	Given an \((r\times n)\)-representation~\(A\)
	for a matroid~$M$
	over a field~\(\mathbb F\),
	a representation of~$M \vee (X,2^X)$
	over~\(\mathbb F\)
	is computable in time of $(n+|X|)(r+|X|)$
	operations
	over~\(\mathbb F\).
\end{lemma}

\begin{proof}
Let \(M=(U,\I)\).
If $U \cap X \not = \emptyset$,
consider the restriction~$M' = (U \setminus X,\{ J\subseteq U\setminus X \mid J \in \I\})$
of~\(M\) to~$U \setminus X$,
which is again a matroid
\citep[Section~1.3]{Oxl92}.
If~\(U\cap X=\emptyset\), then we get $M'=M$.
A linear representation~\(A'\) for~\(M'\)
can be obtained
from a linear representation of~\(A\) for~\(M\)
by removing the columns
labeled by elements in~$X$.
The free matroid~$(X,2^X)$
is representable by the identity matrix
over any field.
Thus,
\(M\vee (X,2^X)=M'\vee (X,2^X)=M'\oplus(X,2^X)\)
is a direct sum
of two matroids
with known linear representations.
The linear representation
of this direct sum
can therefore
be easily obtained
in the claimed time
\citep[Exercise~6, p.~132]{Oxl92}.
\qed\end{proof}

\subsection{Matroids truncations}

\begin{definition}[truncation]\label[definition]{def:trunc}
	The \emph{$k$-truncation}
	of a matroid~$(U, \I)$
	is a matroid~$(U, \I')$
	with~$\I' = \{ S \subseteq U \mid S\in \I\wedge |S| \leq k \}$.
	Moreover,
	if \(A\)~is a linear representation of a matroid
	and \(A'\)~is a linear representation of its
	truncation,
	we will also
	call~\(A'\) a truncation of~\(A\).
\end{definition}

%

\begin{proposition}[{{\citet[Theorem~3.15]{LMPS18}}}]
	\label[proposition]{prop:finite-trunc}
	Let $A$~be an $(r \times n)$-matrix
	of rank~\(r\)
	over a finite field~$\F_{p^d}$,
	where \(p\)~is a prime number
	which is polynomially upper-bounded by 
	the length of the encoding of $A$
	as a binary string.
	For any $k\in\{1,\dots,r\}$,
	we can compute a $k$-truncation of~$A$
	over a finite field extension~$\mathbb K\supseteq\F$
	in time of a polynomial number of field operations over~$\mathbb F$,
	where $\mathbb K = \F_{p^{rkd}}$.
\end{proposition}

\begin{remark}\label[remarkn]{rem:trunc-extension}
	The proof of \citet[Theorem~3.15]{LMPS18}
	shows
	that the field extension~$\mathbb K\supseteq\mathbb F:=\F_{p^d}$
	in \cref{prop:finite-trunc}
	can be chosen as
	$\mathbb K =\F_{p^{sd}}$ for any integer~$s \ge rk$.
	The time for computing
	the truncation
	consists of computing a truncation
	over the field of fractions~\(\mathbb F(X)\)
	using \(O(nkr)\)~operations over~\(\mathbb F\)
	via Theorem~3.14 of \citet{LMPS18}
	and then computing an irreducible
	polynomial of degree~\(s\)
	over~\(\mathbb F\)
	in
	\(s^4d^2\sqrt{p}\cdot\polylog(s,p,d)\)~operations
	over~$\mathbb F$ \citep{Sho90}.
\end{remark}
\cref{prop:finite-trunc}
applied according to \cref{rem:trunc-extension}
immediately yields:

\begin{corollary}
		\label[corollary]{rem:same_field_trunc}
	\looseness=-1
	For \(i\in\{1,\dots,m\}\),
	let $A_i$~be
	$(r_i \times n)$-matrices
	over~$\F_{p^d}$.
	Given a
	natural number~$k \leq \min\{r_i\mid 1\leq i\leq m\}$,
	\(k\)-truncations
	of the~\(A_i\)
	over the same
	finite field extension~$\F_{p^{rkd}}\supseteq\F$
	are computable
	in \(O(mnkr)+r^4k^4d^2\sqrt p\cdot\polylog(r,k,p,d)\)~operations
	over~\(\mathbb F\),
	where $r= \max\{r_i\mid 1\leq i\leq m\}$.
\end{corollary}
Herein,
the additive running time is due to the fact
that we only have to construct the irreducible
polynomial of degree~\(rk\) once
in order to represent the truncated
matroids over the same field extension~\(\F_{p^{rkd}}\).

\section{Representative families for matroid intersections}
\label{sec:tools}

\noindent
Intuitively,
a representative
of some 
set family~\(\cS\)
for a matroid~\(M\) 
is a subfamily~\(\cSs\subseteq\cS\)
such that,
if \(\cS\)~contains a set~\(X\)
that can be extended to
a basis of~\(M\) by adding adding a set~$Y$,
then \(\cSs\) contains a set~$\widehat X$
that can also be extended to a basis of~$M$
by adding~$Y$.
Herein,
the representative~$\cSs$
may be significantly smaller than
the original family~$\cS$,
so that algorithms can gain a speed\hyp up
by working on~$\cSs$ instead of~$\cS$.

\citet{Mar09} first
used representative families 
in randomized
fixed\hyp parameter algorithms
for the NP-hard
\textsc{Matroid Intersection} problem,
where the task is to decide whether
a set is independent in several given matroids.
Representative families have subsequently
been generalized to weighted sets by \citet{FLPS16}
and their construction has been derandomized
by \citet{LMPS18}.
We generalize
this concept to representative
families for matroid \emph{intersections}:

\begin{definition}[\boldmath max intersection $q$-representative family]\label[definition]{def:qrep}
	Given matroids~$\{(U,\I_i)\}_{i=1}^{m}$,
	a~family~$\familyS\subseteq 2^U$,
	and a function~$w\colon \familyS \rightarrow \R$,
	we say that a subfamily~$\cSs \subseteq \familyS$
	is \emph{max intersection $q$-representative
		for $\familyS$ with respect to~\(w\)}
	if,
	for each set~$Y \subseteq U$ of size at most~$q$,
	it holds that:
				if there is a set~$X \in \familyS$  with
		$X \uplus Y \in \bigcap_{i=1}^m\I_i$, 
				then there is a
		set~$\widehat X \in \cSs$ with
		$\widehat X \uplus Y \in \bigcap_{i=1}^m\I_i$ and
		$w(\widehat X) \geq w(X)$.	
	If $m=1$, then we call~$\cSs$
	a \emph{max $q$-representative family} of $\familyS$.
\end{definition}
In this section,
we will show
how to compute max intersection representative families
for  matroids~$\{(U,\I_i)\}_{i=1}^{m}$.
More generally,
for some \(\cH\subseteq 2^U\),
we will compute max intersection representatives for
subsets of the family
\[
  \cB(\cH):=\Bigl\{\biguplus_{j=1}^iH_j\Bigm| i\in\N, \,H_1,\dots,H_i\in\cH\Bigr\}
\]
in a time that will grow merely linearly
in~\(|\cH|\),
whereas the size of~$\cB(\cH)$ is generally
exponential in~\(|\cH|\).
For this to work,
we require the weights of the sets in $\cB(\cH)$
to be computable from weights of sets in~$\cH$
using \emph{inductive union maximizing functions},
which we introduce in \cref{sec:imuf}.
In \cref{sec:rep-unions},
we show how to compute
representatives
with respect to inductive union maximizing set functions
in a single matroid.
In \cref{sec:rep-intersect},
we generalize this result to
max \emph{intersection} representative families
(for multiple matroids).


\subsection{\Dtext{}s}
\label{sec:imuf}


\begin{definition}[\dtext{}]
		\label[definition]{def:dtext}
	Let \(\cH\subseteq 2^U\).
	A set function~\(\dw\colon \cB(\cH)\to\R\)
	is called an \emph{\dtext{}}
	if there is
	a function~\(\dg\colon\R\times \cH\to\R\)
	that is non-decreasing in its first argument and
	such that, for each~\(X\ne\emptyset\),
	\[
	\dw(X)=\smashoperator{\max_{\substack{H\in \cH,S\in\cB(\cH)\\S\uplus H=X}}}\dg(\dw(S),H).
	\]
\end{definition}
An \dtext{}~\(\dw\)
is fully determined
by the value~\(\dw(\emptyset)\)
and the function~\(\dg\).
We thus also say that $\dg$~\emph{generates}~$\dw$.
\Dtext{}s resemble
primitive recursive functions
on natural numbers,
where \(S\uplus H\) plays the role
of the ``successor'' of~\(S\)
in primitive recursion.
We take the maximum over all
partitions~$S\uplus H=X$ since
the partition of~$X$
into a set in~$\cH$ and a set in~$\cB(\cH)$ is not unique.
We now show some examples and counterexamples
for \dtext{}s.

\begin{example}
		\label[example]{ex:ws}
	Let \(\cH\subseteq 2^U\) and
	\(w\colon \cH\to\R\).
        The function~\(\dw\)
	determined by
	\(\dw(\emptyset)=0\) and
	\(\dg\colon (k,H)\mapsto k+w(H)\)
        is an \dtext{}.
	Concretely,
	for \(\emptyset\ne X\subseteq\cB(\cH)\),
	one gets
	\[
	w_{\Sigma}(X):=\dw(X)=\max_{X=H_1\uplus\dots \uplus H_i\atop
		H_1,\dots,H_i \in \cH}\sum_{j=1}^{i}w(H_j)
	\]
	due to the associativity and commutativity
	of the maximum and sum.
\end{example}
\Dtext{}s generalize additive set functions:
\begin{example}
	\label[example]{ex:additive}
	Any additive set function~\(w\colon 2^U\to\R\)
	is \dtex{}
	since,
	for the \dtext{}~\(w_{\Sigma}\)
	in \cref{ex:ws}, one has
	\(w_{\Sigma}(X)=\sum_{j=1}^{i}w(H_j)=w(X)\)
	for \emph{any} partition
	\(X=H_1\uplus\dots \uplus H_j\).
\end{example}
However,
submodular functions
are generally not
\dtex{}:

\begin{example}
		\label[example]{ex:nocoverage}
	Let \(f\colon 2^U\to\R\)~be a coverage function
	(cf.\ \cref{def:covfunc})
	on \(U=\{u_1,v_1,u_2,v_2\}\) with
	\begin{align*}
	u_1&=\{a\}, & v_1&=\{c\}, & u_2 &=\{a,b\},&&\text{and} & v_2=\{c,d\}.
	\end{align*}
	Assume that
	\(f\)~is \dtex{}
	for   \(\cH=\{\{u_1\},\{v_1\},\{u_2\},\{v_2\}\}\).
	The only partition of~\(\{u_1,v_2\}\)
	into sets in~\(\cH\) is~\(\{u_1\}\uplus\{v_2\}\).
	Thus,
	$f(\{u_1,v_2\})=\dg(f(\{u_1\}),\{v_2\})$
	or
	$f(\{u_1,v_2\})=\dg(f(\{v_2\}),\{u_1\})$,
	whichever is larger.
	In the first case, we get the contradiction
	\begin{align*}
	3=f(\{u_1,v_2\})=\dg(f(\{u_1\}),\{v_2\}) = \dg(1,\{v_2\}) =\dg(f(\{v_1\}),\{v_2\})
	\leq f(\{v_1,v_2\})=2.
	\end{align*}
	Otherwise,
	in the second case, we get the contradiction
		\begin{align*}
	3=f(\{u_1,v_2\})=\dg(f(\{v_2\}),\{u_1\})=\dg(1,\{u_1\})=\dg(f(\{u_2\}),\{u_1\})
	\leq f(\{u_1,u_2\})=2.
	\end{align*}
\end{example}
We see that
coverage functions are not
\dtex{} since
the function~\(\dg\)
\dterm{}~\(\dw\) in \cref{def:dtext}
is allowed to
depend only on \(\dw(S)\)~in the first argument,
not on~\(S\) itself.
We will indeed see that this requirement is crucial
and presume that \dtext{}s
are the most general class of functions
with respect to which
we can prove the results in \cref{sec:rep-unions,sec:rep-intersect}.

\subsection{Computing representative families for unions of disjoint sets}
\label{sec:rep-unions}
\noindent
In this section,
we show how to
compute
a representative of subfamilies of~$\cB(\cH)$
with respect to inductive union maximizing set functions
in a single matroid.
We generalize it
to multiple matroids in
\cref{sec:rep-intersect}.

%

\begin{proposition}
  \label[proposition]{thm:core-tool}
	Let $M = (U,\I)$~be a linear matroid
	of rank $r=(\alp+\bet)\gam\geq 1$ with $\alp,\bet\in\mathbb N$, 
	$\cH\subseteq 2^U$~be a $\gam$-family
	of size~$t$, 
	and \(\dw\colon\cB(\cH)\to\R\)~be an \dtext{}
	(cf.\ \cref{def:dtext}) generated by   \(\dw(\emptyset)\)
        and the
	function~\(\dg \colon \R\times \cH\to\R\)
        non\hyp decreasing in the first argument.
	
	Given $\alpha\in\N$,
        a representation~$A$
	of~$M$ over a field~$\mathbb F$,
	the value~\(\dw(\emptyset)\), and the
	function~\(\dg\),
	one can
	compute a
	max $\bet \gam$-representative~$\widehat \familyS$ of 
	size~$r \choose \alp \gam$ for the 
	family
	\[\familyS = \{ S = H_1 \uplus \dots \uplus H_\alp 
	\mid
	S\in\I\text{ and }
	H_j \in \cH\text{ for }j \in \{1,\dots, \alp \}\}
	\]
	with respect to~\(\dw\)
	in time of $O(2^{\omega r}\cdot t)$~operations over $\mathbb F$
	and calls to the
	function~\(\dg\),  
	where $\omega\geq 2$
        is any constant such that two $(n\times n)$-matrices
        can be multiplied in $O(n^{\omega'})$~time for $\omega'<\omega$.
\end{proposition}
Before proving \cref{thm:core-tool},
we provide some context.
The main feature of \cref{thm:core-tool}
is that it allows us
to compute max intersection representatives
of the family~\(\cS\),
whose size may be
exponential in the size of~\(\cH\),
in time growing merely linearly
in the size of~\(\cH\).
The literature
uses several implicit
ad\hyp hoc
proofs of
variants of \cref{thm:core-tool}
in algorithms for concrete problems
\citep{Mar09,FLPS16,FLPS17}.
These proofs usually
use non\hyp negative additive
functions in place of~\(\dw\).
Our \cref{thm:core-tool}
does not require additivity,
yet,
as shown by \cref{ex:additive},
works perfectly fine
for all additive weight functions.

As shown by \cref{ex:nocoverage},
submodular functions
are not necessarily \dtex{}
and,
%
indeed,
we now show that
generalizing \cref{thm:core-tool}
even
to coverage functions
would yield \(\FPT=\Wtwo\).
The proof also makes
for an illustration
of \cref{def:qrep} and \cref{thm:core-tool}:

\begin{observation}
	If \cref{thm:core-tool}
	holds for coverage functions~\(\dw\),
	then \(\FPT=\Wtwo\).
\end{observation}

\begin{proof}
Consider a coverage
function~\(\dw\colon 2^U\to\mathbb N\).
The problem of finding a set
\(S\subseteq U\) with \(|S|\leq r\)
maximizing~\(\dw(S)\)
is known as
\textsc{Maximum Coverage}
and W[2]-hard parameterized
by~\(r\) \citep{CFK+15}.

Now, assume that we could
apply \cref{thm:core-tool}
with a uniform matroid~\(M=(U,\I)\)
of rank~\(r\),
\(\alp=r\), \(\bet=0\), \(\gam=1\),
and
the \(\gam\)-family $\cH = \{ \{ u \} \mid u \in U \}$
to compute a \(0\)-representative~\(\cSs\)
of size \({r\choose\alp\gam}=1\)
of the family
\begin{align*}
\cS=\{S=H_1\uplus\dots\uplus H_\alp\mid S\in\I
\text{ and }H_j\in\cH\text{ for }j\in\{1,\dots,\alp\} \}
=\{S\subseteq U\mid |S|= r\}
\end{align*}
with respect to~\(\dw\)
in time of
\(2^{O(r)}\cdot |\cH|=2^{O(r)}n\)~operations over~\(\F\),
where \(M\)~is representable over
\(\F=\F_{2^d}\)
for \(d=\lceil\log r\rceil\).
Then \(\cSs\)
only contains
the set~\(S\) with \(|S|=r\) maximizing~\(\dw(S)\).
Since each operation over~\(\F_{2^d}\)
can be carried out in \(\poly(d)\in\poly(n)\)~time,
we thus
solve the \(\Wtwo\)-hard
\textsc{Maximum Coverage}
problem
in \(2^{O(r)}\cdot \poly(n)\)~time,
which implies \(\FPT=\Wtwo\).
\qed\end{proof}
We now prove \cref{thm:core-tool}.
Like its implicit ad\hyp hoc proofs
in the literature \citep{Mar09,FLPS16,FLPS17},
we will prove it by iterative
application of the following known result.

\begin{proposition}[{\citet[Theorem~3.7]{FLPS16}}]
	\label[proposition]{prop:repset}
	Let $M = (U,\I)$~be a linear matroid
	of rank~$r=\alp + \bet$,
	$\familyS = \{S_1,\dots,S_t\}$~be
	an $\alp$-family of independent sets,
	and $w \colon \familyS \rightarrow \R$.%
	\footnote{\citet{FLPS16}
		require the weight function to be non\hyp negative.
		Yet their proof does not exploit non\hyp negativity.
		Moreover,
		one can always transform the weight function
		so that it is non\hyp negative and then transform it back.}
	Then,
	there exists a max $\bet$-representative~$\cSs \subseteq \familyS$
	of size at most~${r \choose \alp}$.
	Given a representation of~$M$
	over a field~$\mathbb F$,
	$\cSs$ can be found using
	\[\bigO\left({r \choose \alp}t\alp^{\omega} + t{r \choose \alp}^{\omega-1}\right)\text{ operations over~\(\F\)},\]
        where $\omega$
        is any constant such that two $(n\times n)$-matrices
        can be multiplied in $O(n^\omega)$~time.
\end{proposition}
\cref{alg:q-repset}
computes
the representative families
required by \cref{thm:core-tool}
by iteratively applying \cref{prop:repset}.
\cref{thm:core-tool}
follows from
the correctness proof of
\cref{alg:q-repset} in \cref{lem:q-repset}
and its running time analysis in
\cref{lem:q-repset-running-time}.


%


\begin{algorithm}[t]
  \caption{for the proof of \cref{thm:core-tool}.}
  \label[algorithm]{alg:q-repset}    
  \begin{compactdesc}
  \item[Input:]
    $\alp \in \mathbb N$,
    a $\gam$-family $\cH\subseteq 2^U$ of size~$t$,
    a representation~$A$
    of a matroid~\(M=(U,\I)\) 
    of rank $r=(\alp+\bet)\cdot \gam\geq 1$
    over $\mathbb F$,
    a function~\(\dg\colon \R\times \cH\to\R\)
    non\hyp decreasing in the first argument,
    and \(\dw(\emptyset)\).
    
  \item[Output:] A max $\bet \gam$-representative~\(\widehat\familyS\) of size $r \choose \alp \gam$ for the 
    family $\familyS:=\{ S = H_1 \uplus \dots \uplus H_\alp 
    \mid  S \in \I \text{ and } H_j \in \cH
    \text{ for }j\in\{1,\dots,\alp\}\}
    $ with respect to the
    \dtext{}~$\dw$ generated by \(\dw(\emptyset)\)
    and \(\dg\).
  \end{compactdesc}
  \smallskip\hrule\smallskip
  \begin{compactenum}[\footnotesize 1:]
  \item  $\cSs_0 \gets \{ \emptyset \}$
    \label[line]{lin:base-case}.  

  \item $w_0(\emptyset)\gets \dw(\emptyset)$
    \label[line]{lin:base-case-end}.
    
  \item \label[line]{lin:rs-fork}
    \textbf{for each} $i \in \{ 1, \ldots, \alp \}$ \textbf{do}
    
  \item \quad
    $\cSst_i\gets \emptyset$.\label[line]{lin:csst}
    
  \item \quad
    $w_i\colon 2^U\to\R, S\mapsto -\infty$.\label[line]{lin:wi}
    
  \item \quad \label[line]{lin:for2}
   \textbf{for each} $H\in\cH,S\in\cSs_{i-1},H\cap S=\emptyset,H\cup S\in\I$ \textbf{do}
  \item \qquad
    $\cSst_i\gets \cSst_i\cup\{H\uplus S\}.$ \label[line]{lin:ri}\;
  \item \qquad
    $w_i(H\uplus S)\gets \max\{w_i(H\uplus S),\dg(w_{i-1}(S),H)\}$.  \label[line]{lin:weight}\;

  \item \quad $\cSs_i\gets{}$max $(r-\gam i)$-representative of \(\cSst_{i}\) with respect to~\(w_i\) of size
    at most~\({r\choose \gam i}\) using \cref{prop:repset}.
    \label[line]{lin:compress}
  \item \textbf{return} $\cSs_{\alp}$.\label[line]{lin:retcSsalp}\;
  \end{compactenum}
\end{algorithm}

\begin{lemma}
	\label[lemma]{lem:q-repset}
	\cref{alg:q-repset} is correct.
\end{lemma}
\begin{proof}
We prove by induction that,

\begin{enumerate}[(i)]
\item\label{ind1}
  in \cref{lin:compress} of \cref{alg:q-repset},
  $\cSs_i$ is
	max $(r-\gam i)$-representative     with respect to~\(\dw\)
	for
	\begin{align*}
	\cS_i\phantom:=\{S=H_1\uplus\dots\uplus H_i\mid S\in\I
	\text{ and }H_j\in\cH \text{ for }j\in\{1,\dots,i\}\}.
	\end{align*}
	\item\label{ind2} To this end,
	we simultaneously show
	that
	$w_i(X)=\dw(X)$ for $X\in\cSst_i$.
\end{enumerate}
The lemma then follows since
\cref{alg:q-repset} returns \(\cSs_{\alp}\)
in \cref{lin:retcSsalp},
which has size~\({r\choose\alp \gam{}}\)
by construction in \cref{lin:compress}.

Both \eqref{ind1} and \eqref{ind2} hold
for~\(i=1\)
since \(\cS_1=\cSst_1=\cH\cap \I\)
and
\(\dw(X)=\dg(\dw(\emptyset),X)=w_1(X)\)
for all \(X\in\cH\) by \cref{def:dtext}.
For the induction step,
assume that \eqref{ind1} and \eqref{ind2}
hold for~\(i-1\)   and observe that
\begin{enumerate}[(a)]
	\item by construction,
	\(\cSs_i\subseteq \cSst_i\subseteq \cS_i\)
	for all~\(i\in\{1,\dots,k\}\) and
	\item\label{cla2}  since \(\cH\)~is a \(\gam\)-family
	and every subset of an independent set is independent,
	for all \(X=H\uplus S\in\cS_i\)
	with \(H\in \cH\) and \(S\in\cB(\cH)\),
	one has \(S\in\cS_{i-1}\).
\end{enumerate}
We first prove \eqref{ind2}.
For each set~\(X\) added to~\(\cSst_i\)
by \cref{alg:q-repset}
in \cref{lin:ri},
\begin{align*}
w_i(X)&=\max\{\dg(w_{i-1}(S),H)\mid H\in\cH,S\in\cSs_{i-1},X=H\uplus S\}.\\
\intertext{Since,
	by induction,
	\(w_{i-1}(X)=\dw(X)\)
	for all \(X\in\cSs_{i-1}\subseteq \cSst_{i-1}\),
	we have}
w_i(X)&=\max\{\dg(\dw(S),H)\mid H\in\cH,S\in\cSs_{i-1},X=H\uplus S\}.\\
\intertext{Since,
	by induction,
	\(\cSs_{i-1}\)~is max representative
  for \(\cS_{i-1}\) with respect to~\(\dw\),
  since \(\dg\)~is non\hyp decreasing in its first argument,
  and due to \eqref{cla2},
	we have}
w_i(X)&=\max\{\dg(\dw(S),H)\mid H\in\cH,S\in\cS_{i-1},X=H\uplus S\}\\
&=\max\{\dg(\dw(S),H)\mid H\in\cH,S\in\cB(\cH),X=H\uplus S\}=\dw(X),
\end{align*}
where the last equality is due to \cref{def:dtext}.
We now show \eqref{ind1}.
In \cref{lin:compress},
using \cref{prop:repset},
we create a 
max $(r-\gam i)$-representative $\cSs_i$
of $\cSst_i$ 
with respect to~$w_i$,
which coincides with~\(\dw\) on~\(\cSst_i\)
by \eqref{ind2}.
The claim now is that $\cSs_i$ is
max $(r-\gam i)$-representative
for $\cS_i$
with respect to~$\dw$.

First, let
$Y \subseteq U$ with
$|Y|=r-\gam i$
be such that
there is an~$X \in \cS_i$ with
$Y \uplus X \in \I$.
Since 
\(\dw\)~is
an \dtext{} (cf.\ \cref{def:dtext}),
there is a 
partition~$S \uplus H = X$
with
$H \in \cH$,
$S\in \cB(\cH)$
such that
\begin{equation}
\dw(X)=\dg(\dw(S),H).\label{eq:razlozhenie}
\end{equation}
Since \(X\in\cS_i\), one has \(S\in \cS_{i-1}\) by \eqref{cla2}.
By induction,
$\cSs_{i-1}$~is
max $(r-\gam(i-1))$-representative
for~$\cS_{i-1}$
with respect to~\(\dw\).
Thus,
there is a  set~$S' \in \cSs_{i-1}$
with
$(Y \uplus H) \uplus S' \in \I$, and
$\dw(S')\geq\dw(S)$.
By construction of~\(\cSst_i\) in \cref{lin:ri},
$S' \uplus H \in \mathcal\cSst_i$.
Since,
by \cref{lin:compress},
$\cSs_i$ is max
$(r-\gam i)$-representative
for $\cSst_i$
with respect to~\(w_i\),
there finally is
an~$X' \in \cSs_{i}$ with
$Y \uplus X' \in \I$, and
$\dw(X')=w_i(X')\geq w_i(S'\uplus H)=\dw(S'\uplus H)$.
Since \(\dg\)~is non\hyp decreasing
in its first argument
and by \eqref{eq:razlozhenie},
we get
$\dw(X')\geq \dw(S'\uplus H)
\geq \dg(\dw(S'),H)
\geq \dg(\dw(S),H)
=\dw(X).$

Finally,
consider \(Y\subseteq U\) with
$|Y|<r-\gam i$
such that
there is an~$X \in \cS_i$ with
$Y \uplus X \in \I$.
Since matroid~\(M\) has rank~\(r\),
there is a superset~\(Y'\supseteq Y\)
with \(|Y'|=r-\gam i\)
such that
$Y' \uplus X \in \I$.
Thus, there exists \(X'\in\cSs_{i}\) such that
$Y' \uplus X' \in \I$, and
$\dw(X')\geq \dw(X)$
and both properties hold when
replacing~\(Y'\) by~\(Y\).
\qed\end{proof}
Having shown that \cref{alg:q-repset}
is correct,
to prove \cref{thm:core-tool},
it remains to show that
\cref{alg:q-repset}
runs in the claimed time.

\begin{lemma}\label[lemma]{lem:q-repset-running-time}
	\cref{alg:q-repset} runs in time of
	$O(2^{\omega r} \cdot t)$~operations over~$\mathbb F$
	and calls to the function~\(\dg\),
        where $\omega \geq 2$
        is any constant such that two $(n\times n)$-matrices
        can be multiplied in $O(n^{\omega'})$~time for $\omega'<\omega$.
\end{lemma}
\begin{proof}
	Without loss of generality,
	\(U=\{1,\dots,n\}\).
	We represent subsets of~\(U\)
	as sorted lists.
	Since the input sets in~\(\cH\)
	have cardinality~\(\gamma\),
	we can initially sort each of them
	in \(O(\gamma\log\gamma)\)~time.
	The sorted union and intersection
	of a sorted list of length~\(a\)
	and a sorted list of length~\(b\)
	can be computed in \(O(a+b)\)~time
	\citep[Section~4.4]{AHU83}.
	We thus get a representation
	of sets as words of length~\(r\)
	over an alphabet of size~\(n\).
	We can thus store and look up
	the weight of a set of size at most~\(r\)
	in a trie
	in \(O(r)\)~time \citep[Section~5.3]{AHU83}.
	Note that we do not have the time
	to completely initialize
	the \(O(t)\)~size-$n$ arrays
	in the trie nodes.
	Instead,
	we will initialize
	each array cell of a trie node at
	the first access:
	to keep track of the
	already initialized array positions,
	we use a
	data structure for \emph{sparse sets}
	over a fixed universe~\(U\)
	that allows membership tests,
	insertion, and deletion of elements
	in constant time \citep{BT93}.
	
	%
	
	The running time of \cref{alg:q-repset}
	is dominated by the $\alp\leq r$~iterations 
	of the for-loop starting in \cref{lin:rs-fork}.
	We analyze the running time of iteration~\(i\)
	of this loop.
	The family~$\cH$ consists of $t$~sets of size~$\gam{}$.
	The family~$\cSs_{i-1}$ consists of
	$r \choose {\gam{}(i-1)}$~sets of size $\gam{}(i-1)$ by construction and \cref{prop:repset}.
	Thus,
	the for-loop starting in \cref{lin:for2}
	makes at most \(2^r\cdot t\) iterations:
	\begin{enumerate}[(i)]
		\item $H \cap S = \emptyset$ can be checked
		in $O(\gam{}i)\subseteq O(r)$~time.
		\item We check $H \uplus S \in \I$
		by
		testing \(|H\uplus S|\leq r\)~columns of 
		matrix~\(A\)
		of height~\(r\)
		for linear independence
                in time $O(r^{\omega'})$ \citep{BH74}.
		\item The running time of \cref{lin:ri,lin:weight}
		is dominated
		by
		looking up and storing weights of sets
		of size at most~\(r\)
		in \(O(r)\)~time using a trie,
		and a call to~\(\dg\).
	\end{enumerate}
	Thus,
	the for-loop in \cref{lin:for2}
	runs in time of
	\(O(2^rr^{\omega'} \cdot t)\)~operations over~\(\mathbb F\)
	and calls to~\(\dg\).
	Finally,
	in
	\cref{lin:compress},
	we build a max \((r-\gamma i)\)-representative
	of the \(\gamma i\)-family~\(\cR_i\).
	Since \(|\cR_i|\leq 2^r\cdot t\),
	by \cref{prop:repset},
	this works
	in time of
	\[
	O\left( \left({r \choose \gam{}i}(\gam{}i)^{\omega'} + {r \choose \gam{}i}^{{\omega'}-1}\right)\cdot|\cR_i| \right)
	\subseteq O((2^r r^{\omega'}+2^{r({\omega'}-1)})\cdot 2^r\cdot t) \subseteq O((2^{2r} r^{\omega'}+2^{{\omega'} r})\cdot t)\subseteq O(2^{{\omega'} r} r^{\omega'} \cdot t)
	\]
	operations over~\(\mathbb F\),
        which dominates the running time
        of the for-loop in \cref{lin:for2}.
	Thus,
        \cref{alg:q-repset} runs in time of
	\(O\bigl( r \cdot 2^{{\omega'} r} r^{\omega'} \cdot t  \bigr) 
        \subseteq 
	O( 2^{\omega r} \cdot t)
	\)
      operations over~\(\mathbb F\)
	and calls to~\(\dg\), 
	\qed\end{proof}	

\subsection{Computing intersection representative families}
\label{sec:rep-intersect}
\noindent
In this section,
we generalize \cref{thm:core-tool}
from representatives for a single matroid
to matroid intersections.

\begin{theorem}
	\label[theorem]{lem:max-int-repr}
	Let $\{M_i =(U,\I_i)\}_{i=1}^m$~be
	linear matroids of rank~$r:=(\alp+\bet)\gam\geq 1$,
	\(\cH\subseteq 2^U\) be  a \(\gamma\)-family of size~\(t\),
	and \(\dw\colon\cB(\cH)\to\R\)~be an \dtext{}
	(cf.\ \cref{def:dtext}) generated by   \(\dw(\emptyset)\)
        and the
	function~\(\dg \colon \R\times \cH\to\R\)
        non\hyp decreasing in the first argument.

	
	Given $\alpha\in\N$,
	a
	representation~\(A_i\)
	of~\(M_i\) for each~\(i\in\{1,\dots,m\}\)
	over the same field~$\mathbb F$,
	the value~\(\dw(\emptyset)\), and the
	function~\(\dg\),
	one can
	compute
	a max intersection $\bet\gam$-representative
	of size at most~$r m \choose \alp\gam m$
	of the family
	\[\familyS = \Bigl\{ S = H_1 \uplus \dots \uplus H_\alp 
	\Bigm|
	S\in\bigcap_{i=1}^m \I_i\text{ and }
	H_j \in \cH\text{ for }j \in \{1,\dots, \alp \}\Bigr\}
	\]  
	with respect to~\(\dw\)
	in time of $O(2^{\omega rm}  \cdot (t+n))$~operations
	over $\mathbb F$ and calls to the function~\(\dg\), 
                where $\omega \geq 2$
        is any constant such that two $(n\times n)$-matrices
        can be multiplied in $O(n^{\omega'})$~time for $\omega'<\omega$.
\end{theorem}
%
%
%
%
To prove \cref{lem:max-int-repr},
we reduce the \(m\)~matroid constraints
to a single matroid constraint.
To this end, we use a
folklore construction
sketched by \citet[page~359]{Law76}
in a reduction of the \textsc{Matroid Intersection}
to the \textsc{Matroid Parity} problem.
It works
at the expense of replacing
each universe element
by a ``block'' of \(m\)~copies that
is only allowed to be \emph{completely}
included in or excluded from an independent set.
We then use
our \cref{thm:core-tool}
to compute a representative
of the family of independent sets
that can be obtained as unions of these ``blocks''.
We now present
the folklore construction
and then prove \cref{lem:max-int-repr}.


\begin{lemma}
	\label[lemma]{lem:sum-matroid}
	Let $\{M_i=(U,\I_i)\}_{i=1}^m$~be
	linear matroids of rank $r$ and
	\(\cH\subseteq 2^U\),
	\begin{align*}
	U_{{\oplus}}&{}:=\{ u^{(1)},\ldots, u^{(m)} \mid u \in U \},  \text{ and}\\
	f\colon 2^U\to 2^{U_{\oplus}}, X &{}\mapsto \{ u^{(1)},\ldots, u^{(m)} \mid u \in X \}.
	\end{align*}
	Then,     for all~\(S,S'\subseteq U\),
	\begin{enumerate}[(i)]
		\item \label{sm:1}
                  \(S\ne S'\iff f(S)\ne f(S')\), that is, \(f\)~is injective,
		
		\item \label{sm:2}
		\(f(S)\cup f(S')=f(S\cup S')\),
		\item \label{sm:3}
		$S \cap S' = \emptyset \iff f(S) \cap f(S') = \emptyset$,
	\end{enumerate}
	and given \((r\times n)\)-representations~\(A_i\) of~\(M_i\) for all \(i\in\{1,\dots,m\}\)
	over the same field $\mathbb F$,
	one can,
	in time of $\bigO(m^2 \cdot r \cdot n)$~operations over $\mathbb F$,
	compute a
	\((r m \times n m)\)-representation~$A_\oplus$
	of a matroid~$M_{\oplus}= (U_{\oplus},\I_{\oplus})$ over~$\mathbb F$  such that
	\begin{enumerate}[(i)]
		\setcounter{enumi}{3}
		\item \label{sm:4} for all $S\subseteq U$, $S \in \displaystyle\bigcap\limits_{i=1}^{m} \I_i\iff f(S) \in \I_{\oplus}$.
	\end{enumerate}
\end{lemma}
\begin{proof}
We choose~\(M_{\oplus}\)
to be the
direct sum
of pairwise disjoint copies~\(M_i'\) of~\(M_i\):
\begin{align*}
&M_{\oplus}\phantom:=(U_{\oplus},\I_{\oplus}) = \bigoplus_{i=1}^m M_i',\text{\; where \;} M_i'\phantom:=(U_i,\I_i') \text{\; with}\\
&U_i:= \{ u^{(i)} \mid u \in U \} \text{\; and \;} \I_i' := \{ \{ u_1^{(i)},\ldots,u_j^{(i)} \} \mid \{ u_1,\ldots,u_j \} \in \I_i \}.
\end{align*}
We get a \((r m \times n m)\)-representation~$A_\oplus$
over~$\mathbb F$
of~\(M_{\oplus}\) in time of $\bigO(m^2 \cdot r\cdot n)$~operations 
over~$\mathbb F$ 
\citep[Exercise~6, p.~132]{Oxl92}.

Properties \eqref{sm:1}--\eqref{sm:3}
obviously hold by construction.
It remains to prove \eqref{sm:4}.
Let \(S=\{ u_1,\ldots,u_j \}\subseteq U\)
and,
for an arbitrary $i \in \{1,\ldots, m \}$,
let $S_i = f(S) \cap U_i$.
Then,
$S_i = \{ u_1^{(i)},\ldots,u_j^{(i)}\}$
and
$S \in \I_i$
if and only if
$S_i \in \I_i'$.
Thus,
if $S \in \bigcap_{i=1}^{m} \I_i$,
then  $S_i \in \I_i'$
for all~$i \in \{1,\ldots,m\}$
and, hence,
$\bigcup_{i=1}^{m} S_i = f(S) \in \I_{\oplus}$
by \cref{def:union} on direct sums.
Conversely,
if $f(S) \in \I_{\oplus}$,
then $S_i \in \I_i'$
for all~$i \in \{1,\ldots,m\}$
and, therefore,
$S \in \bigcap_{i=1}^{m} \I_i$.
\qed\end{proof}


\begin{algorithm}[t]
  \caption{for the proof of \cref{lem:max-int-repr}.}
  \label[algorithm]{alg:max-int-repr}
  \begin{compactdesc}
    
  \item[Input:]
    $\alpha\in\N$,
    a \(\gamma\)-family \(\cH\subseteq 2^U\) of size~\(t\),
    representations~$\{A_i\}_{i=1}^m$
    of matroids $\{M_i =(U,\I_i)\}_{i=1}^m$
    of rank~$r:=(\alp+\bet)\gam\geq 1$
    over the same field~$\mathbb F$,
    a function~\(\dg\colon \R\times \cH\to\R\)
    non\hyp decreasing in the first argument,
    and \(\dw(\emptyset)\).
    
  \item[Output:]
    A max intersection $\bet\gam$-representative
    of size at most~$r m \choose \alp\gam m$
    for 
    \(\familyS = \{ S = H_1 \uplus \dots \uplus H_\alp 
    \mid
    S\in\bigcap_{i=1}^m \I_i\text{ and }
    H_j \in \cH\text{ for }j \in \{1,\dots, \alp \}\}
    \)
    w.\,r.\,t.\ to the
    \dtext{}~$\dw$ generated by \(\dw(\emptyset)\)
    and \(\dg\).
  \end{compactdesc}
  \smallskip\hrule\smallskip
  \begin{compactenum}[\footnotesize 1:]
  \item\label[line]{moplus} $A_\oplus\gets$
    representation
    of the
    matroid $M_{{\oplus}}=(U_{{\oplus}},\I_{{\oplus}})$
    created from~\(\{A_i\}_{i=1}^m\) by \cref{lem:sum-matroid}
    over~$\mathbb F$.
  \item \label[line]{hoplus} \(\cH_{\oplus}\gets\{f(X)\mid X\in\cH\}\)
    for the injective function
    $f\colon 2^U \rightarrow 2^{U_{\oplus}}$
    from \cref{lem:sum-matroid}.

    
  \item\label[line]{srep} $\cSs_{\oplus}\gets$
    \cref{alg:q-repset} on~$\alpha$,
    $\gamma m$-family
    $\cH_{\oplus}$,
    $A_\oplus$,
    $\dw_{\oplus}(\emptyset)=\dw(\emptyset)$,
    and $\dg_{\oplus}\colon\R\times\cH_{\oplus}\to\R,
(k,f(H))\mapsto \dg(k,H)$.
  \item\label[line]{sret} \textbf{return} $\{S\subseteq U\mid f(S)\in\cSs_{{\oplus}}\}$.
  \end{compactenum}
\end{algorithm}

\noindent
\cref{lem:max-int-repr} now follows from the following lemma.

\begin{lemma}\label[lemma]{lem:max-int-repr-corr}
  \cref{alg:max-int-repr} is correct and
  runs in time of $O(2^{\omega rm}\cdot (t+n))$
  operations over~$\mathbb F$
  and calls to~\(\dg\).
  If $A_{\oplus}$~in \cref{moplus} is given
  (for example, precomputed),
  then \cref{alg:max-int-repr}
  runs in time of $O(2^{\omega rm}\cdot t)$
  operations over~$\mathbb F$   and calls to~\(\dg\),
  where $\omega \geq 2$
  is any constant such that two $(n\times n)$-matrices
  can be multiplied in $O(n^{\omega'})$~time for $\omega'<\omega$.
\end{lemma}
\begin{proof}
  In \cref{moplus},
  from the linear representations~$\{A_i\}_{i=1}^m$
  of the matroids~\(\{M_i\}_{i=1}^m\),
  \cref{alg:max-int-repr} computes a
  $(rm \times nm)$-representation~\(A_\oplus\)
  of the matroid
  $M_{{\oplus}}=(U_{{\oplus}},\I_{{\oplus}})$
  of rank $rm = (\alp+\bet)\gam m$
  from \cref{lem:sum-matroid}
  in
  time of $O(m^2\cdot r\cdot n)$~operations
  over $\mathbb F$.
  In \cref{hoplus},
  it computes the $\gamma m$-family $\cH_\oplus$.
By \cref{lem:q-repset},
the result of \cref{srep} is
a max \(\beta\gam m\)-representative~\(\cSs_{{\oplus}}\)
of size~$rm \choose \alp \gam m$ for
\begin{align*}
\cS_{{\oplus}}&:=\{f(S)\mid S\in\cS\}
=\{f(S)=f(H_1)\uplus\dots\uplus f(H_\alpha)\mid f(H_i)\in\cH_{\oplus}\text{ and }f(S)\in \I_{\oplus}\},
\end{align*}
where equality is due to \cref{lem:sum-matroid}\eqref{sm:2} and \eqref{sm:4},
with respect to the \dtext{}~\(\dw_{\oplus}\)
determined by
\(\dw_{\oplus}(\emptyset)=\dw(\emptyset)\)
and the
non\hyp decreasing in its first argument function~$\dg_\oplus$.
By \cref{lem:q-repset-running-time},
\cref{srep} is executed
in time of $O(2^{\omega rm}\cdot t)$
operations over~$\mathbb F$
and calls to~\(\dg_{\oplus}\).
Hence,
together with applying the transformation from \cref{lem:sum-matroid},
we take the time
of $O(2^{\omega rm}\cdot (t+n))$
operations over~$\mathbb F$
and calls to~\(\dg_{\oplus}\).
Herein,
one call to~\(\dg_{\oplus}\)
is one call to~\(\dg\).
Also,
since \(\dg_{\oplus}(k,f(H))=\dg(k,H)\)
for all \(k\in\R\) and all \(f(H)\in\cH_{\oplus}\),
one has \(\dw_{\oplus}(f(X))=\dw(X)\)
for all~\(f(X)\in \cB(\cH_{\oplus})\).
This allows us to show that the result
returned in \cref{sret},
\begin{align*}
\cSs&:=\{S\subseteq U\mid f(S)\in\cSs_{{\oplus}}\}=
\Bigl\{H_1\uplus\dots\uplus H_\alp\Bigm|
H_i\in\cH\text{ and }
\biguplus_{i=1}^\alp f(H_i)\in\cSs_{{\oplus}}\Bigr\},
\end{align*}
which
has size~\(|\cSs_{\oplus}|\),
is
max \(\beta \gam\)-intersection
representative of~\(\cS\)
with respect to~\(\dw\).
Note that $\cSs$~can be constructed in
\(\alp\gam m\cdot|\cSs_{\oplus}|\)~time
by simply iterating over
the sets in~\(\cSs_{{\oplus}}\),
replacing
a group of elements~\(u^{(1)},\dots,u^{(m)}\)
by element~\(u\).

To see that $\cSs$ is max \(\beta \gam\)-intersection
representative of~\(\cS\)
with respect to~\(\dw\), consider set $Y \subseteq U$ with \(|Y|\leq\beta\gam\)
and $X \in \mathcal \cS$ with
$Y \uplus X \in \bigcap_{i=1}^m \I_i$.
Then
$f(Y) \uplus f(X) \in \I_{\oplus}$
by \cref{lem:sum-matroid}\eqref{sm:3} and \eqref{sm:4}.
Moreover, $f(X) \in \cS_{\oplus}$.
By \cref{def:qrep},
there is a set $X' \in \cSs_{\oplus}$ such that
$f(Y) \uplus X' \in \I_{\oplus}$,
and $\dw_{\oplus}(X') \geq \dw_{{\oplus}}(f(X))$.
By construction of~\(\cSs\),
there is an \(X''\in\cSs\)
with \(f(X'')=X'\).
Note that \(Y\cap X''=\emptyset\)
by \cref{lem:sum-matroid}\eqref{sm:3}
since \(f(Y)\cap f(X'') = f(Y)\cap X'=\emptyset\).
Moreover,
\(Y\cup X''\in\bigcap_{i=1}^m \I_i\)
by \cref{lem:sum-matroid}\eqref{sm:4}
since \(f(Y)\uplus f(X'')=f(Y)\uplus X'\in \I_{{\oplus}}\).
Finally,
\(\dw(X'')=\dw_{{\oplus}}(f(X''))=\dw_{{\oplus}}(X')\geq \dw_{{\oplus}}(f(X))=\dw(X)\).
\qed\end{proof}


\section{Set packing with multiple matroid constraints}
\label{sec:spmc-fpt}
\label{sec:spmc}



%
\noindent
In this section,
we apply the results from \cref{sec:tools}
to obtain
a fixed\hyp parameter algorithm
for the following problem.


\begin{problem}[\textsc{Set Packing with Matroid Constraints} (\spmc{})]\leavevmode
  \label[problem]{prob:spmc}
	\begin{compactdesc}
		\item[\normalfont\textit{Input:}]
		Matroids $\{(U,\I_i)\}_{i=1}^m$,
		a family $\cH \subseteq 2^U$,
		$w \colon \cH \rightarrow \R$,
		and \(\alp\in\N\).
		\item[\normalfont\textit{Task:}]
		Find sets~\(H_1,\dots,H_\alp\in\cH\) such that
		\[
		\biguplus_{i=1}^\alp H_i\in \bigcap_{i=1}^m \I_i
		\text{\qquad maximizing\qquad}
		\sum_{i=1}^\alp w(H_i).\]
	\end{compactdesc}
\end{problem}
\spmc{} is a generalization of
the \textsc{Matroid Parity}
and \textsc{Matroid Matching} problems introduced
by \citet{Law76}
as generalization of the \textsc{2-Matroid Intersection}
problem.
In \textsc{Matroid Parity} and \textsc{Matroid Matching},
there is only one input matroid and
all input sets in $\cH$ have size exactly two.
In \textsc{Matroid Parity},
all input sets are additionally required to be
pairwise disjoint.
Both problems are solvable in
polynomial\hyp time on linear matroids,
but not on general matroids \cite[Section 43.9]{Sch03}.
\citet{LSV13} studied approximation algorithms
for the variant
\textsc{Matroid Hypergraph Matching}
with one input matroid
and unweighted (but possibly intersecting) input sets.
Finally,
\citet{Mar09} 
and \citet{LMPS18}
obtained fixed\hyp parameter tractability results for
\textsc{Matroid \(\gamma\)-Parity}, 
in which only one matroid is given in the input
and the input set family consists of pairwise non\hyp intersecting unweighted sets
of size~\(\gamma\).
We generalize the fixed\hyp parameter
algorithms of \citet{Mar09} 
and \citet{LMPS18}
to \spmc{}.

\begin{theorem}
  \label[theorem]{thm:solve-spmc}
  \spmc{}
  with sets of size at most~$\gamma$
  and
  $m$~matroids of rank at most~$r$
  with given representations  
  over a field~\(\mathbb F=\mathbb F_{p^d}\)
  is solvable in time
  of
  \(
  2^{O(\alp\gam m)}\cdot |\cH|^2\cdot\poly(r)+ m^2n\cdot\poly(r,\alp,\gam,p,d)
  \)
  operations over~\(\mathbb F\).
\end{theorem}
      
\begin{proof}
  \looseness=-1
  We will prove the theorem using
  \cref{alg:spmc},
  which computes
  \emph{weight} of an optimal solution to SPMC.
  The actual solution
  can then be retrieved via self\hyp reduction,
  calling \cref{alg:spmc} as most \(|\cH|\)~times.
  However,
  note that for the repeated application of \cref{alg:spmc},
  it is enough to compute
  the matroid representations $\{A_i^*\}_{i=1}^m$,
  $\{A_i'\}_{i=1}^m$,
  and $A_{\oplus}$ in
  \cref{alg:max-int-repr} called in \cref{0rep} once,
  as they do not depend on~$\cH$.
  Thus,
  we will only once
  account for the time of their computation
  and analyze the running time
  of $|\cH|$ calls of \cref{alg:spmc}
  under the assumption that they are precomputed.

  \begin{algorithm}[t]
  \caption{for the proof of \cref{thm:solve-spmc}.}
  \label{alg:spmc}
  \begin{compactdesc}
    
  \item[Input:]
    Representations~$\{A_i\}_{i=1}^m$ of
    matroids $\{M_i\}_{i=1}^m$ over a field~$\mathbb F$,
    a family $\cH \subseteq 2^U$
    of sets of size at most~$\gam$,
    $w \colon \cH \rightarrow \R$,
    and $\alpha\in\N$.
  \item[Output:]
    The weight of an optimal solution to SPMC,
    if it exists.
  \end{compactdesc}
  \smallskip\hrule\smallskip
  \begin{compactenum}[\footnotesize 1:]
  \item
    \label[line]{makedummy}
    Create a set~$D=D_1\uplus\dots\uplus D_\alpha$
    of
    $\alp\gam$~\emph{dummy elements}
    with $D \cap U = \emptyset$
    and
    $|D_i|=\gam$
    for $i\in \{1, \dots, \alp\}$.
    
  \item
    \textbf{for each} $H\in \cH$ and $i\in\{1,\dots,\alpha\}$ \textbf{do}
  \item     \label[line]{padh}
    \quad
     $H^{(i)}\gets H$ with additional $\gamma-|H|$ dummy elements chosen arbitrarily from~$D_i$.
    
   \item
     \label[line]{padh2}
     $\cH'\gets \{H^{(i)}\mid H\in\cH, i\in\{1,\dots,\alpha\}\}$.
   \item \textbf{for each} $i\in\{1,\dots,m\}$ \textbf{do}
     
   \item
     \label[line]{padmatroids}
     \quad
     $A_i^*\gets$ linear representation of $M_i^*:= M_i\oplus (D,2^D)$.
     
   \item
     \label[line]{truncatematroids}
     \quad
     $A_i'\gets$ truncation of $A_i^*$ to rank~$\alpha\gamma$.
     
   \item
     \label[line]{0rep}
     \(\cSs\gets\) 
     \cref{alg:max-int-repr}
     with $\alpha$,
     $\gamma$-family~$\cH'$,
     representations~$\{A_i'\}_{i=1}^m$,
     $\dw(\emptyset)=0$
     and \(\dg\colon\R\times\cH'\to\R, (k,H)\mapsto k+w(H\setminus D)\).

   \item \textbf{if} $\cSs=\emptyset$ \textbf{then return} No solution exists.
     
   \item
     \label[line]{retweight}
     \textbf{else return} $w(H\setminus D)$
     for the only set~$H\in\cSs$.
  \end{compactenum}
\end{algorithm}

  First,
  in lines~\ref{makedummy} to \ref{padh2},
  \cref{alg:spmc}  constructs a family~$\cH'$
from~\(\cH\) in which each set has
size exactly~$\gam$.
This step
can be executed in \(|\cH|\alp\gam\)~time.
In \cref{padmatroids},
for each~\(i\in\{1,\dots,m\}\)
and $r_i$~being the rank of~$M_i=(U,\I_i)$,
we compute
a 
representation~\(A_i^*\)
of 
matroid~$M_i^*=(U^*,\I_i^*)$
of rank~\(r_i+|D|=r_i+\alp\gam\)
in
$O((r_i+\alp\gam)(n+\alp\gam))$~operations
over~$\mathbb F$ by \cref{lemma:matroid-set-ext}.
Note that \(|U^*|=n+|D|=n+\alp \gam\) and
\begin{equation}
\begin{aligned}
\text{there are sets $H_1,\ldots,H_\alp \in \cH$}
&
\text{\qquad with $\biguplus_{i=1}^\alp H_i \in \bigcap_{i=1}^m \I_i$
	\quad if and only if}\\
\text{there are sets
	$H_1^{(1)},\ldots,H_\alp^{(\alp)} \in \cH'$}
&
\text{\qquad with $\biguplus_{i=1}^\alp H_i^{(i)} \in \bigcap_{i=1}^m \I_i^*$}.
\end{aligned}
\label{pf4}
\end{equation}
In \cref{truncatematroids},
we compute
$\alp\gam$-truncations~\(\{A_i'\}_{i=1}^m\)
of
the~\(\{A_i^*\}_{i=1}^m\)
using \cref{rem:same_field_trunc}
in time of
\begin{align*}
O(m\cdot (n+\alp\gam)\cdot\alp\gam\cdot(r+\alp\gam))+ (r+\alp\gam)^4\cdot (\alp\gam)^4\cdot d^2\sqrt p\cdot\polylog(r+\alp\gam,p,d)\\
\subseteq mn\cdot\poly(r,\alp,\gam,p,d) \text{ operations over~$\mathbb F$}
\end{align*}
and obtain
an $(\alp\gam \times (n+\alp\gam))$-representation~\(A_i'\)
of \(M_i'\)
over
$\mathbb F' := \mathbb F_{p^{(r+\alp \gam)\alp\gam d}}$
for each~$i \in \{1,\ldots,m\}$.
Note that any operation over~\(\mathbb F'\)
can be executed using
\(\poly(r,\alp,\gam)\)~operations over~\(\mathbb F\).

In \cref{0rep},
we apply
\cref{alg:max-int-repr}
to the linear representations \(\{A_i'\}_{i=1}^m\)
of the matroids~$\{M_i'\}_{i=1}^m$,
all of rank~\(\alp\gam\).
By \cref{lem:max-int-repr-corr},
the result will be  a max 0-representative~\(\cSs\)
of size ${\alp\gam m\choose \alp\gam m}=1$
of the family
\begin{align*}
\cS:=\Bigl\{ S=H_1 \uplus \dots \uplus H_\alp
\Bigm|  S \in \bigcap_{i=1}^m\I'_i\text{ and }H_i \in \cH' \text{ for }i\in\{1,\dots,\alp\}\Bigr\}
\end{align*}
with respect to the \dtext{}
determined by $\dw(\emptyset)=0$
and \(\dg\colon\R\times\cH'\to\R, (k,H)\mapsto k+w(H\setminus D)\).
By \cref{lem:max-int-repr-corr},
this takes time of $2^{O(\alp\gam m)}\cdot |\cH|$~operations
over~$\mathbb F'$ and calls to~$\dg$
if the $A_{\oplus}$ in \cref{alg:max-int-repr}
is precomputed.
The overall running time of solving SPMC
is thus
$|\cH|$~applications of \cref{0rep}
in $2^{O(\alp\gam m)}\cdot \poly(r)\cdot |\cH|^2$~operations
over~$\mathbb F$ and calls to~$\dg$,
plus
$m^2n\cdot\poly(r,\alp,\gam,p,d)$~operations over~$\F$
for computing~$A_{\oplus}$ from \cref{lem:sum-matroid},
the $\{A_i^*\}_{i=1}^m$,
and the $\{A_i'\}_{i=1}^m$ once.
We finally prove that \cref{retweight}
returns the weight of an optimal solution
to our input \spmc{} instance
if and only if that instance
has a feasible solution.
The weight function
on $\cS$
generated by~$\dw(\emptyset)=0$
and $\dg$ is
\begin{equation}\label{eq:w-sigma}
w_{\Sigma}(X):=\max_{X=H_1\uplus\dots \uplus H_\alp\atop
  H_1,\dots,H_\alp \in \cH'}\sum_{i=1}^{\alp}w'(H_i),
\quad
\text{with}
\quad
w'(H_i):=w(H_i\setminus D).
\end{equation}


($\geq$)
Let \(S^*:=\biguplus_{i=1}^\alp H_i\in\bigcap_{i=1}^m\I_i\)
be an optimal solution to \spmc{}.
One has $S:=\biguplus_{i=1}^\alp H_i^{(i)} \in \bigcap_{i=1}^m \I_i^*$
by \eqref{pf4}.
Since \(|S|=\alp\gam\),
\(S\in \bigcap_{i=1}^m \I_i'\).
Thus,
\(S\in\cS\).
Since \(\cSs\) is max 0-representative,
there is an \(S'\in\cSs\)
with
\begin{align*}
w_{\Sigma}(S')&\geq w_{\Sigma}(S)
\geq
\sum_{i=1}^{\alp}
w'(H_i^{(i)})
=\sum_{i=1}^{\alp}
w(H_i).
\end{align*}

($\leq$)
For \(S\in\cSs\),
there are
\(H_1^{(1)},\dots,H_\alp^{(\alp)}\in\cH'\)
such that
\(\biguplus_{i=1}^\alp H_i^{(i)}\in\bigcap_{i=1}^m\I'_i\) and
\[
w_{\Sigma}(S)=
\sum_{i=1}^{\alp}
w'(H_i^{(i)})
=\sum_{i=1}^{\alp}
w(H_i),
\]
which is at most
the weight of an optimal solution to \spmc{},
because the~\(H_i\)
are a feasible solution:
$\biguplus_{i=1}^\alp H_i^{(i)} \in \bigcap_{i=1}^m \I_i^*$.
Thus,
by \eqref{pf4},
$\biguplus_{i=1}^{\alp}H_i\in\bigcap_{i=1}^m\I_i$.
\qed\end{proof}



%
\section{Facility location with multiple matroid constraints}
\label{sec:algorithm}
\noindent
Facility location problems are a classical topic studied
in operations research \citep{LNG15}:
each facility~\(u\) has an opening cost~\(c_u\),
serving client~\(v\) by facility~\(u\) costs~\(p_{uv}\),
and the task is to decide which facilities to open
in order to minimize the total cost
of serving all clients.
\citet{FF11} show that
this problem
is fixed\hyp parameter tractable
when parameterized by an upper bound
on the optimum solution cost,
yet W[2]\hyp hard when parameterized
by the number of facilities that may be opened.
\citet{KKN+14} and \citet{Swa16}
study approximation algorithms
for the problem variant
where the set of facilities
is required to be independent in
a single matroid
and show several applications.
\citet{Kal18} additionally
studies capacity constraints on facilities.

We study facility location problems where not
all clients have to be served,
but where both
the set of clients and the set of facilities 
are required to be independent in multiple matroids.
In this case,
the minimization problem is meaningless:
it would be optimal
not serve any clients and not open any facilities.
Thus,
we study the problem
of \emph{maximizing} income
minus facility opening costs.

\begin{problem}[Uncapacitated facility location with matroid constraints (\mdmc{})]\leavevmode
	\label[problem]{prob:mdmc}
	\begin{compactdesc}
		\item[\normalfont\textit{Input:}]
		A universe~\(U\)
		with \(n:=|U|\),
		for each pair~\(u,v\in U\)
		a \emph{profit} \(p_{uv}\in\mathbb N\)
		obtained when
		a facility at~\(u\) serves a client at~\(v\),
		for each~\(u\in U\)
		a \emph{cost}~\(c_u\in\mathbb N\)
		for opening a facility at~\(u\),
		\emph{facility matroids} \(\{(U_i,A_i)\}_{i=1}^a\),
		and
		\emph{client matroids}~\(\{(V_i,C_i)\}_{i=1}^c\),
		where~\(U_i\cup V_i\subseteq U\).
		\item[\normalfont\textit{Task:}]
		Find two disjoint sets \(A\uplus C\subseteq U\)
		that maximize the \emph{profit}
		\begin{align}
		\sum_{v\in C}\max_{u\in A}p_{uv}-\sum_{u\in A}c_u&&
		\text{such that}&&A&\in\bigcap_{i=1}^a A_i&&\text{and}&&C\in\bigcap_{i=1}^c C_i.\label{mdmcgf}
		\end{align}
	\end{compactdesc}
\end{problem}
%
By convention,
the intersection of
no sets is the whole universe.
Thus,
if $a=0$ or~$c=0$,
this is the same as giving matroids
in which any set of facilities or clients is feasible.
For \mdmc{} without matroid constraints,
\citet{AS99} showed
a 0.828\hyp approximation algorithm
and that there is no polynomial\hyp time
approximation scheme.

\mdmc{} with multiple matroid constraints
can model natural facility location scenarios like
``open at most $k_i$~facilities of type~$i$'' \citep{CLLW16},
even if types are not disjoint,
moving facilities \citep{Swa16},
moving clients,
yet can also model
problems in social network analysis \citep[Section~3]{BTZ19}.
It also generalizes fundamental covering problems:
\begin{example}
	\label[example]{ex:maxcover}
	Using
	\mdmc{}
	with \(a=1\) facility matroid
	and \(c=0\) client matroids,
	one can model
	the classical NP-hard \textsc{Set Cover}
	problem \citep{Kar72}
	of covering a maximum number of elements
	of a set~\(V\)
	using at most \(r\)~sets
	of a collection~\(\cH\subseteq 2^V\).
	To this end,
	choose
	the universe~\(U=V\cup \cH\),
	a single facility matroid~\((\cH,A_1)\)
	with
	\(A_1:=\{H\subseteq \cH\mid |H|\leq r\}\),
	\(c_u=0\) for each \(u\in U\),
	and,
	for each \(u,v\in U\),
	\begin{align*}
	p_{uv}&=
	\begin{cases}
	1&\text{ if $u\in \cH$ such that $v\in u$,}\\
	0&\text{ otherwise}.
	\end{cases}
	\end{align*}
\end{example}
From \cref{ex:maxcover}
and the W[2]-hardness
of \textsc{Set Cover} \citep{CFK+15},
it immediately follows that \mdmc{}
is W[2]-hard parameterized by~$r$
even for zero costs,
binary profits,
and a single uniform facility matroid.
Hence, 
when the set of clients is unconstrained,
the problem of optimally
placing a small number~\(r\)
of \emph{facilities} is hard. 
However,
facility location problems
have also been studied
with a small 
number of \emph{clients} \citep{ABGL15}
and occur in several plausible scenarios \citep{BTZ19}.
We use the tools developed in \cref{sec:tools,sec:spmc}
to analyze the parameterized complexity of \mdmc{}
with a small number of clients.

\begin{theorem}
  \label{thm:alg}
  \mdmc{} is
  \begin{enumerate}[(i)]
  \item \label{thm:hardness}
    W[1]-hard parameterized by~$r$
    for a single client matroid of rank~$r$,
    even with unit costs,
    binary profits,
    and without facility matroids,
  \item\label{thm:alg1}
    solvable in \(2^{\bigO(r\log r)}\cdot n^2\)~time
    for
    a single uniform client matroid
    of rank~\(r\)
    and any
    a single facility matroid
    given as an independence oracle,
  \item\label{thm:alg2}
    fixed\hyp parameter tractable parameterized by \(a+c+r\),
    where \(r\)~is the minimum rank
    of the client matroids
    and representations of all matroids
    over the same finite
    field~\(\F_{p^d}\) are given,
    where $p$~is a prime
    polynomially upper\hyp bounded by the input size.
  \end{enumerate}
\end{theorem}
Before proving \cref{thm:alg},
we put it into context.
\cref{thm:alg}
generalizes and unifies
several fixed\hyp parameter tractability
results from the literature.
\citet{Mar09} showed
that a common independent set
of size~\(r\) in \(m\)~matroids
can be found in \(f(r,m)\cdot\poly(n)\)~time.
Our \cref{thm:alg}\eqref{thm:alg2}
is a direct generalization of this result.
\citet{BPS16} showed that
the problem of covering
at least \(p\)~elements of a set~\(V\)
using at most \(k\)~sets of a given
family~\(\cH\subseteq 2^V\)
is fixed\hyp parameter tractable
parameterized by~\(p\).
This result also follows
from our \cref{thm:alg}\eqref{thm:alg1}
using \cref{ex:maxcover} with an
additional uniform client matroid of rank~\(p\).
Earlier,
\citet{Mar08} showed that
\textsc{Partial Vertex Cover}
(can one cover at least \(p\)~edges of a graph by at most \(k\)~vertices?)
is fixed\hyp parameter tractable by~\(p\).
Our \cref{thm:alg}\eqref{thm:alg1} generalizes this result
and, indeed,
is based on the color coding
approach in \citeauthor{Mar08}'s \citep{Mar08} algorithm.


We now prove \cref{thm:alg}:
\eqref{thm:hardness} is proved in \cref{sec:hardness},
whereas
\eqref{thm:alg1} and \eqref{thm:alg2} are proved
in \cref{sec:cc}.


\subsection{W[1]-hardness for general client matroids
  (proof of \cref{thm:alg}\eqref{thm:hardness})}
\label{sec:hardness}
\noindent
To prove \cref{thm:alg}\eqref{thm:hardness},
we exploit that the
\clique{} problem
is \Wone-hard parameterized by~\(k\)~\cite{CFK+15}:
\begin{problem}[\clique{}]\leavevmode
	\begin{compactdesc}
		\item[\normalfont\textit{Input:}]
		A graph $G=(V,E)$ and integer $k \in \mathbb N$.
		\item[\normalfont\textit{Question:}]
		Does \(G\)~contain a clique on $k$~vertices?
	\end{compactdesc}
\end{problem}
To transfer the \Wone-hardness of \textsc{Clique}
parameterized by~\(k\)
to \mdmc{} parameterized by the client matroid rank,
we reduce the problem
of finding a clique of size~\(k\)
to \mdmc{} with a client matroid of rank~\(2k\).
%
The reduction
is inspired
by the proof
that \textsc{Matroid Parity}
is generally not polynomial\hyp time solvable
when matroids are given using an independence oracle
\citep[Section~43.9]{Sch03}.

%

\begin{construction}
	\label[construction]{constr:hard}
	Let \((G,k)\) with~\(G=(V,E)\)~be an instance
	of \clique{}.
	We construct an instance
	of \mdmc{} as follows.
	For each vertex~$u \in V$,
	let \(D_u:=\{u',u''\}\) such that~\(D_u\cap V=\emptyset\),
	let $L := \{ D_u \mid u \in V  \}$,
	and $S := \bigcup_{ u \in V} D_u$.
	Our \mdmcs{} instance
	consists of the universe~$U = V \uplus S$,
	and, for all \(u,v\in U\),
	\begin{align}
	\label{eq:profit}
	c_u=
	\begin{cases}
	1,&\text{ if $u\in V$,}\\
	0,&\text{ otherwise,}
	\end{cases}
	&&
	p_{uv}=\begin{cases}
	1,&\text{ if $u \in V$ and $v \in D_u$,}\\
	0,&\text{otherwise.}
	\end{cases}
	\end{align}
	We do not use a facility matroid.
	As client matroid, we use the known
	matroid~$(S,\I_C)$ \cite[Section 43.9]{Sch03} with
	\begin{align*}
	\I_C =\{ J \subseteq S \mid{}& |J| \leq 2k -1 \}\\
	{}\cup\{ J \subseteq S \mid{}& |J| = 2k \text{ and } J \text{ is not the union of any $k$~pairs in~$L$} \}\ \\
	{}\cup
	\Bigl\{ J \subseteq S \Bigm|{}&A\subseteq V,|A|=k, J=\smashoperator{\biguplus_{u\in A}}D_{u},\text{ and $G[A]$ is a clique}\Bigr\}.
	\end{align*}
\end{construction}
\begin{proof}[of \cref{thm:alg}\eqref{thm:hardness}]
\cref{constr:hard}
works in polynomial time
and creates an \mdmc{} instance
with unit costs and binary profits as
claimed in \cref{thm:alg}\eqref{thm:hardness}.
Moreover,
the rank of the client matroid is~\(2k\).
We now show that
there is a clique
of size~$k$
in~$G=(V,E)$
if and only if
there is
a solution to the created \mdmc{}
instance with profit~\(k\).

(\(\Rightarrow\))
Let~$A\subseteq V$, \(|A|=k\),
$G[A]$~be a clique,
and $C:=\biguplus_{u\in A}D_{u}$.
Since $|C| = 2k$, 
it follows that $C \in \I_C$.
Hence, \(A\uplus C\)~is a feasible solution to \mdmc{}.
Since
each~$u\in A$
has \(c_u=1\)
and
each~$v\in C$ has \(p_{uv}=1\)
for the \(u\in A\) with \(v\in D_u\),
the profit of~\(A\uplus C\)
given by \eqref{mdmcgf} is
\[\sum_{v \in C} \max_{u \in A} p_{uv} - \sum_{u \in A} c_u = 2k - k = k.\]

(\(\Leftarrow\)) Let~$A \uplus C$~be
an inclusion\hyp minimal
solution with profit at least~\(k\)
to the created \mdmc{} instance.
Since 
each facility~\(u\in A\)
has \(c_u=1\)
and since \(p_{uv}=1\) if and only if~\(v\in D_u\cap C\),
the profit of \(A\uplus C\) given by \eqref{mdmcgf} is
\begin{align}
- \sum_{u \in A} c_u+    \sum_{v \in C} \max_{u \in A} p_{uv}=-|A|+\sum_{u\in A}|D_u\cap C|\geq k.
\label{firstbudget}
\end{align}
Hence,
if there is an \(u\in A\)
with \(|D_u\cap C|\leq 1\),
then \((A \setminus \{u\})\uplus (C\setminus  D_u)\) 
is a solution with at least the same profit,
contradicting the minimality of~\(A\uplus C\).
Thus, for each~\(u\in A\), we have \(D_u\subseteq C\).
Moreover, \(|C|\leq 2k\)
since \(C\in\I_C\) and
the client matroid~\((S,\I_C)\) has rank~\(2k\).
Combining this with \eqref{firstbudget}, we thus get
\begin{align*}
2k\geq\sum_{u\in A}|D_u\cap C|=2|A|\geq k+|A|,
\end{align*}
which means \(|A|=k\).
Thus,
$C = \biguplus_{u\in A}D_u$, \(|C|=2k\), \(C\in\I_C\),
and we conclude that
$G[A]$~is a clique of size~$k$.
\qed\end{proof}

\subsection{Fixed-parameter algorithms
for linear client matroids (proof of \cref{thm:alg}\eqref{thm:alg1} and \eqref{thm:alg2})}
\label{sec:cc}
\noindent
One major difficulty
in solving \mdmc{} is
that the profit
from opening a facility
depends on
which facilities are already open.
To name an extreme example:
when opening only facility~\(u\),
it induces cost~\(c_u\)
and yields profit
from serving all the clients.
However,
when some other facility~\(v\)
is already open,
then additionally opening~\(u\)
induces cost~\(c_u\)
yet might not yield any profit
if all clients are more profitably
already served by~\(v\).
To avoid such interference
between facilities,
we reduce \mdmc{}
to problem variant
with matroid and \emph{color constraints} (\mdmcc{}).

\pagebreak[3]
\begin{problem}[\mdmcc{}]
	\label[problem]{prob:mdmcc}
	\leavevmode
	\begin{compactdesc}
		\item[\normalfont\textit{Input:}]
		A universe~\(U\),
		a coloring~\(\col\colon U\to\{1,\dots,k+\ell\}\),
		a partition
		\(Z_1\uplus\dots\uplus Z_\ell=\{\ell+1,\dots,\ell+k\}\),
		for each pair~\(u,v\in U\)
		a \emph{profit}~\(p_{uv}\in\mathbb N\)
		gained when
		a facility at~\(u\) serves a client at~\(v\),
		for each~\(u\in U\)
		a \emph{cost}~\(c_u\in\mathbb N\)
		for opening a facility at~\(u\),
		\emph{facility matroids}~\(\{(U_i,A_i)\}_{i=1}^a\),
		and
		\emph{client matroids} \(\{(V_i,C_i)\}_{i=1}^c\),
		where~\(U_i\cup V_i\subseteq U\).
		\item[\normalfont\textit{Task:}]
		Find two sets \(A\uplus C\subseteq U\) such that
		\begin{enumerate}[(i)]
			\item\label{mdmcc:colored-fac} for each \(i\in\{1,\dots,\ell\}\),
			there is exactly  one facility \(u\in A\) with \(\col(u)=i\), 
			\item \label{mdmcc:colored-client} for each \(i\in\{\ell+1,\dots,\ell+k\}\),
			there is exactly one client \(v\in C\)
			with \(\col(v)=i\),
			\item\label{mdmcc:independent} \(
			\displaystyle
			A\in\bigcap_{i=1}^a A_i\text{\quad and\quad}C\in\bigcap_{i=1}^c C_i,
			\)
		\end{enumerate}
		and that maximizes
		\begin{align}
		\sum_{u\in A}\Bigl(-c_u+
		\smashoperator{\sum_{v\in C\cap Z(u)}}p_{uv}
		\Bigr),
		\label{mdmccgf}
		\text{ where }
		Z(u):=\{v\in U\mid
		\col(v)\in Z_{\col(u)}\}.
		\end{align}
	\end{compactdesc}
\end{problem}
We call $A\uplus C$
a \emph{feasible solution}
if it satisfies~\eqref{mdmcc:colored-fac}--\eqref{mdmcc:independent},
not necessarily maximizing~\eqref{mdmccgf}.
In \mdmcc{},
a facility~\(u\) of color~\(i\)
gets profit
only from clients in~\(Z(u)\),
that is, from clients
with a color in~\(Z_i\).
Moreover,
there can be only one facility of color~\(i\)
and the clients with a color in~\(Z_i\)
are only served by facilities of color~\(i\).
Thus,
the contribution of a facility~\(u\)
to the goal function of \mdmcc{}
is independent
from the contributions of other facilities,
which makes \mdmcc{} significantly
easier than \mdmc{}.
We prove the following lemma.
\begin{lemma}
	\label[lemma]{lem:turingred}
	If \mdmcc{} is solvable
	in \(t(k+\ell)\)~time,
	then
	\mdmc{}
	is solvable in
	\(2^{\bigO(r\log r)}(t(2r)+n)\log n\)~time,
	where \(r\)~is the minimum rank
	of the client matroids.
\end{lemma}
To prove \cref{lem:turingred},
we will prove that
\begin{enumerate}[(a)]
	\item\label{cc-dir1}
	any feasible
	solution to \mdmcc{}
	is a feasible solution
	with at least the same profit
	for \mdmc{}
	and
	
	\item\label{cc-dir2}
	for any optimal solution~\(A\uplus C\)
	to \mdmc{},
	we can sufficiently quickly find
	a coloring
	and
	a partition~\(Z_1\uplus\dots\uplus Z_\ell\)
	such that
	\(A\uplus C\)~is a feasible
	solution
	with at least the same profit
	for \mdmcc{}.
\end{enumerate}
Towards \eqref{cc-dir2},
note that
we cannot simply try out
all partitions and colorings;
for example,
there are
\(\bigO(r^n)\)~colorings---%
too many for \cref{lem:turingred}.
However,
since there is a client matroid
of rank~\(r\),
there is an optimal solution
to \mdmc{}
such that \(|C|=\ka\)
and \(|A|=\ell\)
for some \(\ell\leq \ka\leq r\).
Thus,
without loss of generality
assuming \(U=\{1,\dots,n\}\),
%
the colorings in a
\emph{$(n,\ka+\ell)$-perfect hash family}
as defined below
will contain a coloring
such that the elements
of \(A\uplus C\) get pairwise distinct colors:
\begin{definition}[\boldmath $(n,s)$-perfect hash family
	{\citep[Definition~5.17]{CFK+15}}]
	An \emph{\((n,s)\)-perfect hash family}
	is a set~\(\mathcal F\)
	of functions~\(f\colon\{1,\dots,n\}\to\{1,\dots,s\}\)
	such that,
	for any \(S\subseteq\{1,\dots,n\}\)
	with \(|S|\leq s\),
	there is a function~\(f\in \mathcal F\)
	injective on~\(S\).
\end{definition}

\begin{proposition}[{{\citep[Theorem~5.18]{CFK+15}}}]
	\label[proposition]{prop:perfhash}
	An \((n,s)\)-perfect hash family
	of size \(e^ss^{\bigO(\log s)}\log n\)
        can be computed
	in
	\(e^ss^{\bigO(\log s)}n\log n\)~time.
\end{proposition}
In the following,
we will prove that
\cref{alg:turingred}
correctly solves \mdmc{}.

\begin{algorithm}
  \caption{for the proof of \cref{lem:turingred}.}
  \label[algorithm]{alg:turingred}
  \begin{compactdesc}
    
  \item[Input:]
    An \mdmc{} instance:
    universe~\(U=\{1,\dots,n\}\),
    profits~\(p_{uv}\in\N\)
    for each~$u,v\in U$,
    costs~\(c_u\in\N\)
    for each~$u\in U$,
    facility matroids \(\{(U_i,A_i)\}_{i=1}^a\),
    and client matroids \(\{(V_i,C_i)\}_{i=1}^c\)
    of minimum rank~\(r\),
    where \(U_i\cup V_i\subseteq U\).
    
  \item[Output:]
    An optimal solution~\(A\uplus C\) to \mdmc{}.
  \end{compactdesc}
    \smallskip\hrule\smallskip
  \begin{compactenum}[\footnotesize 1:]
  \item \label[line]{lin:guessell} \textbf{for each} $1\leq\ell\leq\ka\leq r$  \textbf{do}
  \item \quad    \(\mathcal F\gets(n,\ell+\ka)\)-perfect hash family. \label[line]{lin:perfhash}
  \item \quad \label[line]{lin:forhash} \textbf{for each} $f\in\mathcal F$ \textbf{do}
    
  \item \qquad \label[line]{lin:x} \textbf{for each} $X\subseteq\{1,\dots,\ell+k\}$ with \(|X|=\ell\) \textbf{do}
  \item  \label[line]{lin:startlinear}%
    \qquad\quad  bijectively rename colors in~\(f\) and~\(X\) so that \(X=\{1,\dots,\ell\}\).
        
  \item \qquad\quad \label[line]{lin:partitions} \textbf{for each} partition \(Z_1\uplus\dots\uplus Z_\ell=\{\ell+1,\dots,\ell+\ka\}\) \textbf{do}
  \item \qquad\qquad Solve \mdmcc{} with            universe~\(U\),
    coloring~\(f\),
    partition~\(Z_1\uplus\dots\uplus Z_\ell\),

    \qquad\qquad\qquad
    profits~\((p_{uv})_{u,v\in U}\),
    costs~\((c_u)_{u\in U}\),
    matroids \(\{(U_i,A_i)\}_{i=1}^a\)
    and \(\{(V_i,C_i)\}_{i=1}^c\)%
    \label[line]{lin:solve}.
    
  \item \textbf{return} maximum\hyp profit \mdmcc{} solution found in any iteration.
    \label[line]{lin:return}\;
  \end{compactenum}
\end{algorithm}

\begin{lemma}
	\label[lemma]{lem:mdmcc-to-mdmc}
	Any \mdmcc{} solution~\(A\uplus C\)
	is also an \mdmc{} solution
	with at least the same profit.
\end{lemma}

\begin{proof}
Let \(A\uplus C\) be
any \mdmcc{} solution.
Obviously,
it is also feasible for \mdmc{}.
Its profit as an \mdmcc{} solution is
given by \eqref{mdmccgf}
and,
denoting
\(\bar p_{uv}:=p_{uv}\)
if $v\in Z(u)$ and \(\bar p_{uv}:=0\) otherwise,
is
$$\sum_{u\in A}\Bigl(-c_u+
\smashoperator{\sum_{v\in C\cap Z(u)}}p_{uv}
\Bigr)
={}\sum_{u\in A}\Bigl(-c_u+\sum_{v\in C}\bar p_{uv}\Bigr)
={}-\sum_{u\in A}c_u+\sum_{v\in C}\sum_{u\in A}
\bar p_{uv}.$$
Since, for each~\(v\in C\),
	there is exactly one~\(u\in A\)
	with \(v\in Z(u)\),
	this is
$$\leq
-\sum_{u\in A}c_u+\sum_{v\in C}\max_{u\in A} p_{uv},
$$
which is exactly the cost of~\(A\uplus C\)
as a solution to \mdmc{} as given by \eqref{mdmcgf}.
\qed\end{proof}

\begin{lemma}
	\label[lemma]{lem:algcorrect}
	\cref{alg:turingred} is correct.
\end{lemma}

\begin{proof}
Let \(A\uplus C\)~be
an optimal solution to \mdmc{}
such that \(A\)~is of minimum size.
We show that \cref{alg:turingred}
outputs a solution to \mdmc{}
of equal profit.

Since one of the client matroids has rank~\(r\),
one has \(|A|=\ell\) and~\(|C|=\ka\)
such that \(1\leq \ell\leq\ka\leq r\).
\cref{alg:turingred}
tries these \(\ell\) and~\(\ka\) in
\cref{lin:guessell}.
Thus,
the \((n,\ka+\ell)\)-perfect hash family~\(\mathcal F\)
generated in \cref{lin:perfhash}
contains a function~\(\col\colon U\to\{1,\dots,\ka+\ell\}\)
that is bijective restricted to~\(A\uplus C\).
\cref{alg:turingred}
tries this function~\(\col\) in \cref{lin:forhash}.
Since \(|A|=\ell\)
contains elements of pairwise
distinct colors,
\cref{alg:turingred}
in \cref{lin:x}
iterates over the color set~\(X\) of~\(A\)
and renames all colors so that \(X=\{1,\dots,\ell\}\).
We get that \(A\)~contains
exactly one element of each color of \(X=\{1,\dots,\ell\}\)
and that \(C\)~contains
exactly one element of each color of~\(\{\ell+1,\dots,\ell+k\}\).

Now,
recall that \(U=\{1,\dots,n\}\)
and,
for each \(v\in C\), let \(m(v)\in A\)~be
the facility with minimum index
that serves~\(v\) with maximum profit,
that is,
\[m(v):=\min\{u\in A\mid p_{uv}=\max_{w\in A}p_{wv}\}.\]
Then,
for each~\(u\in A\),
there is a \(v\in C\)
such that \(m(v)=u\):
otherwise,
we can rewrite the goal function \eqref{mdmcgf}
of \mdmc{}
as
\[
\sum_{v\in C}\max_{u\in A}p_{uv}-\sum_{u\in A}c_u=
\sum_{v\in C}p_{m(v),v}-\sum_{u\in A}c_u,
\]
and removing~\(u\) from~\(A\)
would yield a solution
to \mdmc{} with at least the same profit but smaller~\(A\),
contradicting the minimality of~\(A\).
Thus,
for each \(i\in\{1,\dots,\ell\}\),
the set~\(Z_i\)~of colors of the clients~\(v\)
served by a facility~\(m(v)\) of color~\(i\),
that is,
\begin{align*}
Z_i&:=\{\col(v)\mid v\in C, m(v)=u, \col(u)=i\},
\end{align*}
is nonempty. Hence,
\(Z_1\uplus\dots\uplus Z_\ell=\{\ell+1,\dots,\ell+\ka\}\)
is a partition:
\begin{itemize}
	\item Equality follows since \(C\)~contains exactly one element of each color of~\(\{\ell+1,\dots,\ell+k\}\)
	and \(m(v)\)~is defined for each~\(v\in C\).
	\item Pairwise disjointness follows
	since \(m(v)\)~for each client~\(v\in C\) is unique.
\end{itemize}
Since \(Z_1\uplus\dots\uplus Z_\ell\)~is a partition,
\cref{alg:turingred} in \cref{lin:partitions}
iterates over this partition
and  
\(A\uplus C\)~is a feasible solution
to the \mdmcc{} instance in this iteration.
We show that its profit as an \mdmcc{}
solution,
given by \eqref{mdmccgf},
is the same as the profit as an \mdmc{} solution,
given by \eqref{mdmcgf}.
To this end,
denote
\[
\bar p_{uv}:=
\begin{cases}
p_{uv}&\text{ if $m(v)=u$, and}\\
0&\text{ otherwise},
\end{cases}
\]
and observe that,
for \(u\in A\) and~\(v\in C\),
one has
\(m(v)=u\) if and only if
\(\col(v)\in Z_{\col(u)}\).
By choice in \eqref{mdmccgf},
this is if and only if
\(v\in C\cap Z(u)\).
Thus, the cost of~\(A\uplus C\)
as a solution to \mdmc{} is
\begin{align*}
\sum_{v\in C}\max_{u\in A}p_{uv}-\sum_{u\in A}c_u
  =\sum_{v\in C}\sum_{u\in A}\bar p_{uv}-\sum_{u\in A}c_u
=\sum_{u\in A}\sum_{v\in C\cap Z(u)} p_{uv}-\sum_{u\in A}c_u,
\end{align*}
which is exactly the profit~\eqref{mdmccgf}
of~\(A\uplus C\)
as a solution to the \mdmcc{} instance
given solved in this iteration.
Thus,
in \cref{lin:return},
\cref{alg:turingred}
will return an \mdmcc{}
solution with at least this profit.
By \cref{lem:mdmcc-to-mdmc},
this will be an \mdmc{} solution
of at least the same profit.
Since \(A\uplus C\)~is
an optimal \mdmc{} solution,
we conclude that the solution
returned by \cref{alg:turingred}
is also optimal.
\qed\end{proof}
We can now complete
the reduction of \mdmc{} to \mdmcc{}.

\begin{proof}[of \cref{lem:turingred}]
We have shown in \cref{lem:algcorrect}
that \cref{alg:turingred}
correctly solves \mdmc{}.
It remains to analyze the running time.

The loop in \cref{lin:guessell}
makes \(r\cdot(r+1)/2\)~iterations.
Observe that
\(\ell+k\leq 2r\).
By \cref{prop:perfhash},
the perfect hash family
in \cref{lin:perfhash}
is computable in
\(e^{\ell+\ka}\cdot(\ell+\ka)^{\bigO(\log(\ell+\ka))}\cdot
n\log n\)~time
and
the loop in \cref{lin:forhash}
makes
\(e^{\ell+\ka}\cdot(\ell+\ka)^{\bigO(\log(\ell+\ka))}\cdot\log n\)~iterations.
\Cref{lin:startlinear} works in \(\bigO(n)\)~time,
whereas \cref{lin:solve} works
in \(t(k+\ell)\)~time
by assumption.
There are at most \(\ell^{\ell+k}\)~variants
to chose \(X\uplus Z_1\uplus \dots\uplus Z_{\ell}\).
Thus, the overall running time of the algorithm is
\(2^{\bigO(r\log r)}(t(2r)+n)\log n\).
\qed\end{proof}


\subsubsection{One arbitrary facility matroid and
	one uniform client matroid}
\label{sec:uniclients}
\noindent
We now prove \cref{thm:alg}\eqref{thm:alg1}:
an algorithm that solves
\mdmc{} in \(2^{\bigO(r\log r)}\cdot n^2\)~time
if there is only one,
yet arbitrary (not necessarily linear)
facility matroid
and one uniform client matroid of rank~\(r\).
To this end,
we show:
\begin{proposition}
	\label[proposition]{lem:mdmccpoly}
	\mdmcc{} is solvable in \(O(\ell n^2)\)~time
	if there is one facility matroid
	given as an independence oracle
	and one client matroid
	that is uniform.
\end{proposition}
Then,
\cref{thm:alg}\eqref{thm:alg1}
follows from \cref{lem:mdmccpoly,lem:turingred}.
To present the algorithm
for the proof of \cref{lem:mdmccpoly},
we introduce the following notation:
\begin{definition}
		\label[definition]{def:ui}
	For a universe~\(U\)
	with coloring~\(\col\colon\allowbreak U\to\{1,\dots,k+\ell\}\),
	we denote by $U(i):=\{u\in U\mid \col(u)=i\}$
	the elements of color~\(i\).
\end{definition}
To prove \cref{lem:mdmccpoly},
we use \cref{alg:mdmccpoly},
which solves \mdmcc{} as follows.
In \cref{lin:fu},
for each facility~\(u\),
it computes a set~\(F(u)\)
containing of each color in~\(Z_{\col(u)}\)
exactly one client~\(v\in V_1\)
that maximizes~\(p_{uv}\).
The intuition is that
if facility~\(u\) will be part of a solution,
then the clients~\(F(u)\) will
\emph{follow}~\(u\) into the solution.
In \cref{lin:wu},
it assigns
to each facility~\(u\)
a weight~\(w(u)\),
which is the profit
gained from serving
the clients in~\(F(u)\) by~\(u\)
minus the cost for opening facility~\(u\).
Finally,
in \cref{lin:compA},
it computes a maximum\hyp
weight set~\(A\in A_1\)
containing exactly one facility
of each color~\(\{1,\dots,\ell\}\)
and chooses \(C:=\bigcup_{u\in A}F(u)\).
The crucial point herein is
that the set~\(A\) can be computed
as the
maximum\hyp weight common independent set
of size~\(\ell\)
of two matroids,
which can be done in polynomial time
\citep[Sections~41.3 and~41.3a]{Sch03}.
In the following,
we prove the correctness
and the running time of
\cref{alg:mdmccpoly}.
\begin{algorithm}
  \caption{for the proof of \cref{lem:mdmccpoly}.}
  \label[algorithm]{alg:mdmccpoly}
  \begin{compactdesc}
    
  \item[Input:]
    An \mdmcc{} instance:
    a universe~\(U=\{1,\dots,n\}\),
    profits~\(p_{uv}\in\N\) for each~$u,v\in U$,
    costs~\(c_u\in\N\) for each~$u\in U$,
    a~coloring~\(\col\colon U\to\{1,\dots,k+\ell\}\),
    a facility matroid~\((U_1,A_1)\),
    a uniform client matroid \((V_1,C_1)\) of rank~\(r\),
    where~\(U_1\cup V_1\subseteq U\), and a
    partition~\(Z_1\uplus\dots \uplus Z_\ell=\{\ell+1,\dots,\ell+k\}\).
    
  \item[Output:] An optimal solution~\(A\uplus C\) to \mdmcc{}.
  \end{compactdesc}
  \smallskip\hrule\smallskip
  \begin{compactenum}[\footnotesize 1:]
  \item \label[line]{lin:check-kr} \textbf{if} $k>r$ \textbf{then return} No solution exists.
  \item \label[line]{lin:checkcol} \textbf{if} $\{\ell+1,\dots,\ell+k\}\nsubseteq \{\col(v)\mid v\in V_1\}$ \textbf{then return} No solution exists.
  \item \textbf{for each} $u\in U_1$ \textbf{do}
  \item       \label[line]{lin:fu}\quad
    \(\displaystyle
    F(u)\gets\Bigl\{\argmax_{v\in V_1\cap U(j)}p_{uv}\Bigm| j\in Z_{\col(u)}\Bigr\}\).
    
  \item \quad \(w(u)\gets -c_u+\sum_{v\in F(u)}p_{uv}\). \label[line]{lin:wu}

  \item     Compute $A\in A_1$ 
    maximizing \(\sum_{u\in A}w(u)\)
    and
    containing exactly one element of each color in~\(\{1,\dots,\ell\}\).
    \label[line]{lin:compA}\;
    
  \item  \label[line]{lin:check-A}
    \textbf{if} \(A\) not found
    \textbf{then return} No solution exists.
      \label[line]{lin:nosolution}
  \item     \(C\gets\bigcup_{u\in A}F(u)\).
    \label[line]{lin:compC}\;
  \item \textbf{return} $A\uplus C$.  \label[line]{lin:retAC}\;
  \end{compactenum}
\end{algorithm}

\begin{lemma}
	\label[lemma]{lem:mdmccpoly-correct}
	\cref{alg:mdmccpoly} is correct.
\end{lemma}

\begin{proof}
First, assume that
\cref{alg:mdmccpoly}
returns some \(A\uplus C\)
in \cref{lin:retAC}.
We show that
\(A\uplus C\)~is
a feasible solution to
the input \mdmcc{} instance.
By construction of~\(A\)
in \cref{lin:compA},
\(A\in A_1\) and contains exactly one element
of each color~\(\{1,\dots,\ell\}\).
Thus,
\cref{prob:mdmcc}\eqref{mdmcc:colored-fac}
is satisfied.
Since \(Z_1\uplus\dots\uplus Z_\ell\)~is
a partition of~\(\{\ell+1,\dots,\ell+k\}\)
and
\cref{lin:checkcol}
has been passed,
\(C=\bigcup_{u\in A}F(u)\)
for the sets~\(F(u)\subseteq V_1\) computed
in \cref{lin:fu}
contains 
exactly one element
of each color~\(\{\ell+1,\dots,\ell+k\}\)
and \eqref{mdmcc:colored-client} is satisfied.
Thus, \(C\subseteq V_1\) and \(|C|=k\).
Moreover,
\(k\leq r\)
since \cref{lin:check-kr}
has been passed.
Thus,
since \((V_1,C_1)\)~is a uniform
matroid of rank~\(r\),
it follows that \(C\in C_1\)
and \eqref{mdmcc:independent} is satisfied.
We conclude that
\(A\uplus C\)~is a feasible solution.

Now assume that there is
an optimal solution~\(A^*\uplus C^*\)
to \mdmcc{}.
We show that
\cref{alg:mdmccpoly}
returns a solution with the same profit.
First,
since \(C^*\)~contains
exactly one vertex of each color in~\(\{\ell+1,\dots,\ell+k\}\)
by \cref{prob:mdmcc}\eqref{mdmcc:colored-client},
we get \(|C^*|=k\).
Second,
since \(C^*\in C_1\) by \eqref{mdmcc:independent},
we get \(C^*\subseteq V_1\) and \(k=|C^*|\leq r\).
Thus,
the tests in \cref{lin:check-kr,lin:checkcol}
pass.
Thus,
\cref{lin:compA}
of
\cref{alg:mdmccpoly}
computes a set~\(A\in A_1\)
containing exactly
one element of each color~\(\{1,\dots,\ell\}\)
(by \eqref{mdmcc:colored-fac},
\(A^*\)~witnesses the existence of such a set),
a corresponding set~\(C=\bigcup_{u\in A}F(u)\)
in \cref{lin:compC},
and finally
returns \(A\uplus C\) in \cref{lin:retAC},
which we already proved to be a feasible solution
for \mdmcc{}.
It remains to
compare the profit of~\(A\uplus C\)
to that of~\(A^*\uplus C^*\).
To this end,
the goal function \eqref{mdmccgf} for~\(A^*\uplus C^*\)
can be rewritten as

$$
\sum_{u\in A^*}w'(u)\text{ \quad for \quad}w'(u):=-c_u+\smashoperator{\sum_{v\in C^*\cap Z(u)}}p_{uv}.$$
In comparison,
	consider the weight~\(w(u)\)
	assigned to each \(u\in A^*\)
	as in \cref{lin:wu} of \cref{alg:mdmccpoly}.
	Since, for each~\(u\in A\) with \(\col(u)=i\),
$$ w(u)=-c_u+\smashoperator{\sum_{v\in F(u)}}p_{uv}=-c_u+\sum_{j\in Z_i}\max_{v\in V_1\cap U(j)}p_{uv},$$
one has \(w(u)\geq w'(u)\).
Since \(A\)~computed in \cref{lin:compA}
maximizes \(\sum_{u\in A}w(u)\),
\begin{align*}
\sum_{u\in A^*}w'(u)
\leq \sum_{u\in A}w(u)
=\sum_{u\in A}
\Bigl(
-c_u+\smashoperator{\sum_{v\in F(u)}}p_{uv}
\Bigr)
=\sum_{u\in A}
\Bigl(
-c_u+\smashoperator{\sum_{v\in C\cap Z(u)}}p_{uv}
\Bigr),
\end{align*}
which is exactly the profit
of solution~\(A\uplus C\) to \mdmcc{}.
\qed\end{proof}
Having shown the correctness of \cref{alg:mdmccpoly},
we now analyze its running time.

\begin{lemma}
	\label[lemma]{lem:mdmccpoly-time}
	\cref{alg:mdmccpoly} can be run in \(\bigO(\ell n^2)\)~time
	if the matroid \((U_1,A_1)\)
	is given as an independence oracle.
\end{lemma}

\begin{proof}
Lines~\ref{lin:check-kr} to \ref{lin:wu},
\ref{lin:compC}, and \ref{lin:retAC}
are easy to implement in \(\bigO(n^2)\)~time.
We show how to execute \cref{lin:compA} efficiently.
To this end,
consider the partition
matroid~\((U_1',B)\)
in which a subset of~\(U_1\)
is independent
if it contains at most one element
of each color~\(\{1,\dots,\ell\}\)
and no elements of other colors.
\Cref{lin:compA} is then
computing a set~\(A\) of maximum weight
and cardinality~\(\ell\)
that is independent
in both matroids~\((U_1,A_1)\) and~\((U_1,B)\).
This can be done in
\(\bigO(\ell n^2)\)~time
\citep[Sections~41.3 and~41.3a]{Sch03}.
\qed\end{proof}
\cref{lem:mdmccpoly-time,lem:mdmccpoly-correct,lem:turingred}
together finish
the proof of \cref{thm:alg}\eqref{thm:alg1}.

\subsubsection{Facility and client matroids representable over the same field}
\label{sec:alggeneral}
\noindent
In this section,
we prove \cref{thm:alg}\eqref{thm:alg2}: \mdmc{} is fixed\hyp parameter tractable parameterized by the number of matroids and the minimum rank over all client matroids
if all matroids are representable
over the same field.
To this end,
we prove the following,
which, together with \cref{lem:turingred},
yields \cref{thm:alg}\eqref{thm:alg2}.
\begin{proposition}
	\label[proposition]{lem:mdmccfpt}
	An optimal solution to an \mdmcc{} instance~\(\I\)
	can be found in
	$f(a+c+k+\ell)\cdot \poly(|\I|)$~time
	if representations
	of the \(a\)~facility matroids
	and \(c\)~client matroids
	over the same field~\(\F_{p^d}\) are given
	for some prime~\(p\) polynomially upper\hyp bounded
        by~\(|\I|\).
\end{proposition}
The algorithm for \cref{lem:mdmccfpt}
is more involved
than \cref{alg:mdmccpoly},
which breaks in
the presence of client matroids,
even a single one:
we cannot guarantee
that the sets~\(F(u)\)
chosen in \cref{lin:fu} of \cref{alg:mdmccpoly}
are independent in the client matroids
or that their union will be.
Ideally,
one would we able to choose
from all possible subsets~\(F(u)\subseteq Z(u)\)
of clients that could be served by~\(u\),
yet there are too many.
Here the \emph{max intersection representative families}
that we construct in \cref{sec:rep-intersect}
come into play:
using \cref{lem:max-int-repr},
we compute a family~\(\cFs(u)\)
so that, if there is any set of clients
that can be served by~\(u\)
and that is independent in all client
matroids together with the clients
served by other facilities,
then \(\cFs(u)\)~contains
at least one such subset
yielding at least the
same profit.
Using \cref{thm:solve-spmc},
we can then compute
disjoint unions of these
sets maximizing profit.
To describe the algorithm,
we introduce some notation.
\begin{definition}
	For a universe~\(U\)
	with coloring~\(\col\colon\allowbreak U\to\{1,\dots,k+\ell\}\) and \(U(i)\) as in \cref{def:ui},
	we denote by
	\begin{align*}\allowdisplaybreaks
	U_A&:=\smashoperator{\bigcup_{i=1}^\ell} U(i)&&\text{ is the set of \emph{facilities}, and}\\
	U_C&:=\smashoperator{\bigcup_{i=\ell+1}^{\ell+k}} U(i)&&\text{ is the set of \emph{clients}.}
	\end{align*}
\end{definition}
\cref{alg:mdmccfpt}
now solves \mdmcc{} as follows.
In \cref{alg2:color-matroid},
it constructs 
a multicolored matroid~\(M_P\)
that will ensure that
any independent
set of \(k\)~facilities
and \(\ell\)~clients
fulfills
\cref{prob:mdmcc}\eqref{mdmcc:colored-fac}
and \eqref{mdmcc:colored-client}.
In \cref{alg2:mat-ext},
it computes a family~\(\mathcal M\)
of matroids that contains~\(M_P\)
and all facility and client matroids,
which
are extended 
so that
a set~$A \uplus C \subseteq U$
is independent in all of them
if and only if $A$~is independent
in all facility matroids
and $C$~is independent in all client matroids.
Now,
if one of the matroids in~\(\mathcal M\) has rank
less than~\(k+\ell\),
then there is no common independent
set of \(\ell\)~facilities
and \(k\)~clients,
which is checked in \cref{alg2:smallrank}.
The truncation in
\cref{alg2:truncation}
thus results in each matroid in~\(\mathcal M\)
having rank exactly~\(k+\ell\),
which is needed
to apply \cref{lem:max-int-repr}
in
\cref{alg2:get-followers}.
In \cref{alg2:get-followers},
we construct for each~$u \in U_A$
with \(\col(u)=i\)
a max intersection~$(k + \ell - |Z_i|)$-representative~$\cFs(u)$
for the family~\(\cF(u)\)
of all sets of clients
that could potentially be served by~$u$
in a solution.
Afterwards, in \cref{alg2:the-family},
we construct a family of sets,
each consisting of one facility~$u \in U_A$
and a potential client set from~$\cFs(u)$.
Finally,
in \cref{alg2:ask-marx-for-help},
we will use \cref{thm:solve-spmc}
to combine $\ell$~of such sets
into a set that is independent in all matroids
in~$\mathcal M$ and yields maximum profit.
To prove \cref{lem:mdmccfpt},
we now show that
\cref{alg:mdmccfpt}
is correct and
analyze its running time.

\begin{algorithm}
  \caption{for the proof of \cref{lem:mdmccfpt}.}
  \label[algorithm]{alg:mdmccfpt}
  \begin{compactdesc}
  \item[Input:]
    An \mdmcc{} instance:
    universe~\(U=\{1,\dots,n\}\),
    partition~\(Z_1\uplus\dots \uplus Z_\ell=\allowbreak\{\ell+1,\dots,\ell+k\}\),
    coloring~\(\col\colon U\to\{1,\dots,k+\ell\}\),
    profits~\(p_{uv}\in\N\)
    for each~$u,v\in U$,
    costs~\(c_u\in\N\)
    for each~$u\in U$,
    facility matroids $\mathcal A = \{ (U_i,A_i) \}_{i=1}^a$,
    client matroids $\mathcal C = \{ (V_i,C_i) \}_{i=1}^c$,
    all given as representations over the same finite field,
    where $U_i,V_i \subseteq U$.
    
  \item[Output:]
    An optimal solution~\(A\uplus C\) to \mdmcc.
  \end{compactdesc}
    \smallskip\hrule\smallskip
  \begin{compactenum}[\footnotesize 1:]
  \item   \label[line]{alg2:color-matroid}
    \(M_P\gets (U,\{ I \subseteq U \mid\)  \(I\)~has at most one element  of each color in \(\{1,\dots,k+\ell\})\).
  \item      \label[line]{alg2:mat-ext}
    \(\mathcal M\gets\{M_P\}\cup\{M\vee (U_C,2^{U_C})\mid M\in\mathcal A\} \cup \{M\vee (U_A,2^{U_A})\mid M\in\mathcal C\}\).
  \item     \label[line]{alg2:smallrank}
    \textbf{if}
    any matroid in \(\mathcal M\) has rank
    less than \(k+\ell\)
    \textbf{then return}
    No solution exists.
    
  \item Truncate all matroids in~$\mathcal M$
    to rank~$k+\ell$ (using \cref{rem:same_field_trunc}).%
    \label[line]{alg2:truncation}\;
  \item \label[line]{alg2:for-loop}
    \textbf{for each} $u \in U_A$ and \(i:=\col(u)\) \textbf{do}
    
  \item \quad
      \(\widehat\cF(u)\gets\)
      max intersection~$(k+\ell-|Z_{i}|)$-representative for the family
      \begin{align*}
        \cF(u):=\{ I \subseteq Z(u)\mid
        I \text{ is independent in each of } 
        \mathcal M
        \text{ and } |I|=|Z_i|\}
      \end{align*}
      \quad with respect to	weights~$w_u\colon 2^U\to\N,I\mapsto\sum_{v\in I}p_{uv}$ (via \cref{lem:max-int-repr}).
      \label[line]{alg2:get-followers}
      
    \item \quad
      $\cFs[u] \gets \{ X \cup \{ u \} \mid X \in \cFs(u) \}$.
      \label[line]{alg2:pad-followers}

    \item 
    $\cFs \gets \bigcup_{u\in U_A}\cFs[u]$.
    \label[line]{alg2:the-family}\;
  \item     \(S_1,\dots,S_\ell\gets\) solution to
    \spmc{}  with matroids~\(\mathcal M\),
    family~\(\cFs\), and
    weights $w\colon \cFs\to\Z,X\mapsto w_u(X\setminus\{u\})-c_u$, where $\{u\}=X\cap U_A$ (via \cref{thm:solve-spmc}).
    \label[line]{alg2:ask-marx-for-help}
  \item \textbf{if} not found \textbf{then return}
    No solution exists.
  \item     $A \gets U_A \cap (S_1 \cup \dots \cup S_\ell)$.

  \item $C \gets U_C \cap (S_1 \cup \dots \cup S_\ell)$.
  \item \textbf{return} $A \uplus C$.
    \label[line]{alg2:retsol}
  \end{compactenum}
\end{algorithm}

\begin{lemma}
		\label[lemma]{lem:alg2-feasible}
	Any solution output by \cref{alg:mdmccfpt}
	is feasible for \mdmcc{}.
\end{lemma}

\begin{proof}
If \cref{alg:mdmccfpt}
outputs a solution in \cref{alg2:retsol},
then, in \cref{alg2:ask-marx-for-help},
it founds
sets~$S_1,\dots,S_{\ell} \in \cFs$
such that \(S=S_1\uplus\dots\uplus S_{\ell}\)
is independent in all matroids~$\mathcal M$.
We show that \(A\uplus C\)
for
$A = S \cap U_A$
and
$C = S \cap U_C$ is a feasible solution for \mdmcc{},
that is, it satisfies
properties
\cref{prob:mdmcc}\eqref{mdmcc:colored-fac}--\eqref{mdmcc:independent}.

\eqref{mdmcc:colored-fac} and \eqref{mdmcc:colored-client}:
Observe that each set in~$\cFs$
(constructed in \cref{alg2:the-family})
contains
exactly one facility~$u \in U_A$
and \(|Z_i|\)~elements from~\(U_C\)
for \(i=\col(u)\).
Thus,
\[
|S|=\sum_{i=1}^\ell|S_i|=\ell+\sum_{i=1}^\ell|Z_i|=\ell+k
\]
since \(Z_1\uplus\dots\uplus Z_\ell=\{\ell+1,\dots,\ell+k\}\).
Since \(S\)~is independent
in the multicolored matroid~$M_P$,
it follows that \(S\)~contains
exactly one facility and exactly one client
of each color.
Since \(A=S\cap U_A\) and \(C=S\cap U_C\),
\eqref{mdmcc:colored-fac} and \eqref{mdmcc:colored-client}
hold.

\eqref{mdmcc:independent}:
Since \(A\)~is independent in
all matroids of~\(\mathcal M\),
it is independent in all matroids
of~\(\{M\vee (U_C,2^{U_C})\mid M\in\mathcal A\}\).
Since \(A\subseteq U_A\) and thus
\(A\cap U_C=\emptyset\),
it follows that \(A\)~is independent
in all matroids in~\(\mathcal A\).
Analogously,
it follows that \(C\)~is independent
in all matroids in~\(\mathcal C\).
\qed\end{proof}

\begin{lemma}
	\label[lemma]{lemma:alg2-correct}
	Given a feasible
	\mdmcc{} instance,
	\cref{alg:mdmccfpt} outputs
	a solution of maximum profit.
\end{lemma}

\begin{proof}
Let \(S=A\uplus C\)
be an optimal solution to \mdmcc{}.
We show that \cref{alg:mdmccfpt}
outputs a solution of \mdmcc{}
with the same profit.

Since \(S\)~contains exactly one
facility of each color in~\(\{1,\dots,\ell\}\)
by \cref{prob:mdmcc}\eqref{mdmcc:colored-fac}
and exactly one client
of each color in~\(\{\ell+1,\dots,\ell+k\}\)
by \eqref{mdmcc:colored-client},
it is independent in the colorful matroid~\(M_P\)
constructed in \cref{alg2:color-matroid}.
Moreover, by \eqref{mdmcc:independent},
\(C\)~is independent in all matroids in~\(\mathcal C\)
and trivially in~\((U_C,2^{U_C})\).
Similarly, \(A\)~is independent in all matroids
in~\(\mathcal A\) and~\((U_A,2^{U_A})\).
Thus,
by \cref{def:matunion} about matroid unions,
\(A\uplus C\)~is independent in
all matroids in the set~\(\mathcal M\)
constructed in \cref{alg2:mat-ext}.
Since \(|A\uplus C|=k+\ell\),
it follows that each matroid in~\(\mathcal M\)
has rank at least~\(k+\ell\)
and \cref{alg2:smallrank} is passed.
It follows that after \cref{alg2:truncation},
all matroids in~\(\mathcal M\)
have rank \emph{exactly} \(k+\ell\).

Now,
consider an arbitrary facility~\(u\in A\)
and \(i:=\col(u)\).
For the set~\(Z(u)\) in \eqref{mdmccgf}
and the set~\(\cF(u)\) constructed in
\cref{alg2:get-followers},
one has \(C_u:=C\cap Z(u)\in\cF(u)\)
and
\[
w_{u}(C_u)=\sum_{v\in C_u}p_{uv}.
\]
Moreover, one has
\(|A\uplus C|=k+\ell\),
\(|C_u|=|Z_i|\),
and
\[
A\uplus C=A\uplus\biguplus_{w\in A}C_{w}=(A\uplus\smashoperator{\biguplus_{w\in A\setminus\{u\}}}C_{w})\uplus C_u.
\]
Since \(\cFs(u)\) is max intersection
\((k+\ell-|Z_i|)\)-representative
with respect to~\(w_{u}\),
by \cref{def:qrep},
there is \(C_u'\in\cFs(u)\) with
\(w_{u}(C_u')\geq w_{u}(C_u)\)
and such that
\[
A\uplus C'=(A\uplus\smashoperator{\biguplus_{w\in A\setminus\{u\}}}C_{w})\uplus C_u'
\]
is independent in all matroids of~\(\mathcal M\).
Consequently,
\(\{u\}\uplus C_u'\in\cFs\) in \cref{alg2:the-family}
and
\[
A\uplus C''=\biguplus_{u\in A}(\{u\}\uplus C_u')
\]
is a feasible solution to the \spmc{} instance
in \cref{alg2:ask-marx-for-help}.
Thus,
\cref{alg:mdmccfpt} in \cref{alg2:ask-marx-for-help}
finds an optimal \spmc{} solution~\(S_1,\dots,S_\ell\in\cFs\)
for the weights $w\colon \cFs\to\Z,X\mapsto w_u(X\setminus\{u\})-c_u$, where \(\{u\}=X\cap U_A\).
It returns \(A^*\uplus C^*\) for
\(A^*=U_A\cap(S_1\uplus\dots\uplus S_\ell)\)
and
\(C^*=U_C\cap(S_1\uplus\dots\uplus S_\ell)\)
in \cref{alg2:retsol},
which is a feasible solution for \mdmcc{} by
\cref{lem:alg2-feasible}.
Finally,
since each such set~\(S_i\in\cFs\)
consists of one facility~\(u\in U_A\)
and \(|Z_{\col(u)}|\)~elements of~\(Z(u)\)
with pairwise distinct colors,
the profit of \(A^*\uplus C^*\)
as a solution to \mdmcc{}
given by \eqref{mdmccgf} is
\begin{align*}
&\sum_{u\in A^*}(-c_u+\smashoperator{\sum_{v\in C^*\cap Z(u)}}p_{uv})=
\sum_{j=1}^\ell w(S_j)
\geq \sum_{u\in A} w(\{u\}\uplus C_u')\\
&=\sum_{u\in A}(-c_u+w_u(C_u'))\geq \sum_{u\in A}(-c_u+w_u(C_u))
=    \sum_{u\in A}(-c_u+\sum_{v\in C_u}p_{uv}),
\end{align*}
which is exactly the profit of the optimal
solution~$A\uplus C$.
\qed\end{proof}

\begin{lemma}
	\label[lemma]{lemma:alg2-runtime}
	Given representations of all matroids over the same field $\mathbb F$, where \fieldconst{the input size},
	\cref{alg:mdmccfpt} can be executed in
	$2^{O(\ell k (a+c))}\cdot\poly(x)$~time,
	where $a$ is the number of facility matroids,
	$c$ is the number of client matroids, $k+\ell$ is the number of colors, 
	and $x$ is the input size.
\end{lemma}
\begin{proof}
First,
we compute
a $((k+\ell) \times n)$-representation~\(B=(b_{ij})\)
of the multicolored matroid~$M_P$ in
\cref{alg2:color-matroid} of 
\cref{alg:mdmccfpt} over~\(\mathbb F\)
in $\bigO(nk)$~time:
\(b_{ij}=1\) if
element $j\in U$ has color~$i$, and
\(b_{ij}=0\) otherwise.
By \cref{lemma:matroid-set-ext},
we can compute the set~$\mathcal M$
of matroids and their representations over~\(\mathbb F\)
in \cref{alg2:mat-ext}
in time of a polynomial number
of field operations over~$\F$.

Due to \cref{alg2:smallrank},
all matroids in~\(\mathcal M\) have rank at
least~\(k+\ell\).
In \cref{alg2:truncation},
we use \cref{rem:same_field_trunc}
to compute $(k+\ell)$-truncations
of all matroids in~$\mathcal M$
over a field extension~$\mathbb F' \supseteq\mathbb F$,
in a polynomial number
of field operations over~$\mathbb F$.
Herein,
\(\mathbb F'=\F_{p^{d'}}\)
with
$d'= (k+\ell)\cdot rd \in \poly(n+d)$ since $\ell\leq k \le n$
and $r$~is the maximum rank of the input matroids.

Let $m=|\mathcal M|=a+c+1$.
Using \cref{lem:max-int-repr} with~\(\gamma=1\),
\(\cH:=\{\{v\}\mid v \in \cF(u)\}\),
and weight function~\(w_{u}\colon 2^U\to\N\)
to implement \cref{alg2:get-followers},
we can execute the
for-loop starting at \cref{alg2:for-loop}
in time of $2^{\bigO((k+\ell)\cdot m)} \cdot n^2$ 
operations over $\F'$:
since \(\cH\)~is a 1-family
and,
thus,
the partition of any
subset of~\(U\) into sets of~\(\cH\) is unique,
the function~\(w_{u}\)
is a \dtext{}
generated by \(w_{u}(\emptyset)=0\)
and the constant\hyp time computable
function~\(\dg\colon \N\times \cH\to\N,(k,\{v\})\mapsto k+p_{uv}\)
(cf.\ \cref{ex:ws}).

In \cref{alg2:pad-followers},
for each $u\in U_A$,
it holds that $|\cFs[u]|\leq|\cFs(u)| \leq {{(k + \ell)m} \choose {km}}\leq 2^{(k+\ell)m}$
by \cref{lem:max-int-repr}.
Thus, in \cref{alg2:the-family},
we have
\(|\cFs|\leq n 2^{(k+\ell)m}\).
Moreover,
each set in~\(\cFs\) has size at most~\(k+1\).
Therefore,
by \cref{thm:solve-spmc},
\cref{alg2:ask-marx-for-help}
can be executed in time of
\(2^{O(\ell (k+1) m)}\cdot n\cdot\poly(p,d')\)
operations over~\(\mathbb F'\).

Since we initially get the representations
of the input matroids 
over the field~$\mathbb F_{p^d}$
with $p^d$ elements, we need
at least $d \log p$ bits to encode
an element of the field.
Thus, $d \log p$ is less than
the input size.
Therefore, \(d\) and \(d'\) are polynomially
bounded by the input size.
Since
each element of the field $\mathbb F_{p^{d'}}$ 
can be encoded using $d' \log p$ bits,
each field operation over~$\mathbb F'$ 
(and therefore over~$\mathbb F$)
can be executed in $\poly(d' \log p)$~time, 
which is polynomial in the input size.
Thus,
\cref{alg:mdmccfpt}
can be executed in
$2^{O(\ell k (a+c))}\cdot\poly(x)$~time.
\qed\end{proof}
\cref{lem:mdmccfpt}
now follows from
\cref{lemma:alg2-correct,lemma:alg2-runtime,lem:alg2-feasible}.
Finally,
\cref{thm:alg}\eqref{thm:alg2} follows
from \cref{lem:mdmccfpt,lem:turingred}.

\section{Conclusion}
\noindent
The complexity of \mdmc{}
seems to be determined
by the client matroids:
it is fixed\hyp parameter tractable
parameterized by the minimum rank
of the client matroids
in case when
the facility matroid is arbitrary
and the client matroid is uniform,
or when all matroids are linear.
The problem becomes W[1]-hard
for general client matroids,
even without facility matroids.
It would be interesting to settle
the complexity of \mdmc{}
with one \emph{arbitrary} facility matroid
parameterized by the rank
of a single linear client matroid.

We point out that the algorithms
in \cref{thm:alg}\eqref{thm:alg1} and \eqref{thm:alg2}
are easy to implement:
the construction of perfect
hash families using \cref{prop:perfhash}
can be replaced by coloring the universe
uniformly at random with~\(k+\ell\) colors \citep{AYZ95}
and the truncation
of matroids using \cref{rem:same_field_trunc},
involving large field extensions
and generation of irreducible polynomials,
can be replaced by a very simple
randomized algorithm 
that does not enlarge fields \citep[Proposition~3.7]{Mar09}.
Doing so,
when aiming for an error probability
of at most~\(\varepsilon\in(0,1)\),
the asymptotic running time of our algorithms
increases by a factor~\(\ln(1/\varepsilon)\).

For future research,
we point out that our algorithm
for \cref{thm:alg}\eqref{thm:alg1}
works in polynomial space,
whereas \cref{thm:alg}\eqref{thm:alg2}
requires exponential space
due to \cref{lem:max-int-repr,thm:solve-spmc}.
It is interesting whether this is avoidable.
Moreover,
given that approximation algorithms
are known for UFLP
without matroid constraints~\citep{AS99},
for the minimization variant of UFLP
with a single facility matroid~\citep{KKN+14,Swa16},
as well as for
other optimization problems
under matroid constraints \citep{CCPV11,FW12,LSV13},
it is canonical to study
approximation algorithms
for \mdmc{}.


\paragraph{Acknowledgments}

We thank F.\ V.\ Fomin, F.\ Panolan,
and the anonymous referees
for valuable input.
This study was initiated
at the 7th annual research retreat
of the Algorithmics and Computational Complexity group
of TU Berlin,
Darlingerode, Germany,
March 18th--23rd, 2018.

\paragraph{Funding}
René van Bevern
was supported by Russian Foundation for Basic Research (RFBR)
grant~18-501-12031~NNIO\textunderscore a.
Oxana Yu.\ Tsidulko
was supported by RFBR
grant 18-31-00470 mol\_a.

\bibliographystyle{mydam}
\bibliography{mcmc}

\end{document}